\newif\ifDRAFT
\DRAFTfalse

\newif\ifSUB
\SUBfalse

\newif\ifCONF
\CONFtrue

\newif\ifJOURNAL
\JOURNALfalse

\ifSUB
\documentclass[11pt]{article}
\usepackage{fullpage}
\usepackage{url}
\usepackage{amsmath}
\usepackage{mathtools}
\usepackage{amssymb}
\usepackage{graphicx}

\else\ifCONF
\documentclass[11pt]{article}
\usepackage{fullpage}
\usepackage{url}
\usepackage{amsmath}
\usepackage{mathtools}
\usepackage{amssymb}
\usepackage{graphicx}

\else\ifJOURNAL
\documentclass[prodmode,acmjacm]{acmsmall} % Aptara syntax

\usepackage[ruled]{algorithm2e}

\SetAlFnt{\small}
\SetAlCapFnt{\small}
\SetAlCapNameFnt{\small}
\SetAlCapHSkip{0pt}
\IncMargin{-\parindent}

\acmVolume{0}
\acmNumber{0}
\acmArticle{0}
\acmYear{2045}
\acmMonth{0}
\usepackage{url}
\usepackage{amsmath}
\usepackage{amssymb}
\usepackage{graphicx}
\fi\fi\fi

\usepackage{stmaryrd}
\usepackage{enumerate}
\usepackage[colorlinks,urlcolor=blue,citecolor=blue,linkcolor=blue]{hyperref}
\ifDRAFT
\usepackage[inline]{showlabels} %after hyperref and amsmath
\fi

\newcommand{\normF}[1]{{\| #1 \|}_F}
\newcommand{\norm}[1]{\lVert #1 \rVert}

\DeclareMathOperator{\nnz}{\mathtt{nnz}}
\DeclareMathOperator{\diag}{\mathtt{diag}}
\DeclareMathOperator{\colspace}{\mathtt{colspace}}
\DeclareMathOperator{\range}{\mathtt{range}}
\DeclareMathOperator{\rowspan}{\mathtt{rowspan}}
\DeclareMathOperator{\tr}{\mathtt{tr}}
\DeclareMathOperator{\rank}{\mathtt{rank}}
\DeclareMathOperator*{\argmin}{argmin}
\DeclareMathOperator{\sr}{\mathtt{sr}}
\DeclareMathOperator{\sd}{\mathtt{sd}}

\newcommand\piloi{{\tt piloi}}
\newcommand\loi{{\tt loi}}
\newcommand\roi{{\tt roi}}
\newcommand\piroi{{\tt piroi}}

\newcommand\tO{\tilde{O}}
\newcommand\tA{\tilde{A}}
\newcommand\ty{\tilde{y}}
\newcommand\hA{\hat{A}}

\newcommand\hb{\hat{b}}
\newcommand\hB{\hat{B}}

\newcommand\tx{\tilde{x}}

\newcommand\tX{\tilde{X}}
\newcommand\tY{\tilde{Y}}
\newcommand\tW{\tilde{W}}
\newcommand\tH{\tilde{H}}
\newcommand\tZ{\tilde{Z}}
\newcommand\hS{\hat{S}}
\newcommand\hR{\hat{R}}
\newcommand\Iden{{I}}

\newcommand\twomat[2]{\left[\begin{smallmatrix} #1 \\ #2 \end{smallmatrix} \right] }
\newcommand\twomatr[2]{\left[\begin{smallmatrix} #1 & #2 \end{smallmatrix} \right] }
\newcommand\nm[1]{\norm{\cdot}_{#1} }

\newcommand{\myvertiii}[1]{{\vert\kern-0.25ex\vert\kern-0.25ex\vert #1
    \vert\kern-0.25ex\vert\kern-0.25ex\vert}}

\ifDRAFT
\newcommand{\marrow}{\marginpar[\hfill$\longrightarrow$]{$\longleftarrow$}}
\newcommand{\niceremark}[3]
   {\textcolor{red}{\textsc{#1 #2:} \marrow\textsf{#3}}}
\newcommand{\niceremarkblue}[3]
   {\textcolor{blue}{\textsc{#1 #2:} \marrow\textsf{#3}}}
\newcommand{\Ken}[2][says]{\niceremark{Ken}{#1}{#2}}
\newcommand{\David}[2][says]{\niceremark{David}{#1}{#2}}
\newcommand{\Haim}[2][says]{\niceremarkblue{Haim}{#1}{#2}}
\usepackage[inline]{showlabels}
\else
\newcommand{\Ken}[1]{}
\newcommand{\David}[1]{}
\newcommand{\Haim}[1]{}
\fi

\newcommand{\poly}{{\mathrm{poly}}}
\newcommand{\eps}{\varepsilon}
\newcommand{\R}{{\mathbb R}}

\ifJOURNAL
\newtheorem{fact}[theorem]{Fact}
\else
\newtheorem{theorem}{Theorem}

\newtheorem{definition}[theorem]{Definition}
\newtheorem{lemma}[theorem]{Lemma}

\newtheorem{corollary}[theorem]{Corollary}

\newtheorem{fact}[theorem]{Fact}

\newtheorem{remk}[theorem]{Remark}

\newenvironment{remark}{\begin{remk}
\begin{normalfont}}{\end{normalfont}
\end{remk}}

\fi %ifJOURNAL

\def\FullBox{\hbox{\vrule width 8pt height 8pt depth 0pt}}

\def\qed{\ifmmode\qquad\FullBox\else{\unskip\nobreak\hfil
\penalty50\hskip1em\null\nobreak\hfil\FullBox
\parfillskip=0pt\finalhyphendemerits=0\endgraf}\fi}

\def\qedsketch{\ifmmode\Box\else{\unskip\nobreak\hfil
\penalty50\hskip1em\null\nobreak\hfil$\Box$
\parfillskip=0pt\finalhyphendemerits=0\endgraf}\fi}

\newenvironment{proof}{\begin{trivlist} \item {\bf Proof:~~}}
  {\qed\end{trivlist}}

\begin{document}

\title{Sharper Bounds for Regularized Data Fitting}

\author{Haim Avron\thanks{Tel Aviv University} \\
\and
Kenneth L. Clarkson\thanks{IBM Research - Almaden} \\
\and
David P. Woodruff\thanks{IBM Research - Almaden}}
\date{}

\maketitle

\begin{abstract}

We study matrix sketching methods for regularized variants of linear regression,
low rank approximation, and canonical correlation analysis. Our main focus
is on sketching techniques which preserve the objective function value
for regularized problems, which is an area
that has remained largely unexplored. 
We study regularization
both in a fairly broad setting, and in the specific context of the popular and widely
used technique of ridge regularization; for the latter, as applied to each of these
problems, we show algorithmic resource bounds in which the
{\em statistical dimension} appears in places where in previous bounds
the rank would appear.
The statistical dimension is always smaller than the rank, and decreases as
the amount of regularization increases. In particular, for the ridge low-rank
approximation problem
$\min_{Y,X} \norm{YX - A}_F^2 + \lambda\norm{Y}_F^2 + \lambda\norm{X}_F^2$,
where $Y\in\R^{n\times k}$ and $X\in\R^{k\times d}$,
we give an approximation algorithm needing 
$O(\nnz(A)) + \tO((n+d)\eps^{-1}k \min\{k, \eps^{-1}\sd_\lambda(Y^*)\})+ \poly(\sd_\lambda(Y^*) \epsilon^{-1})$ 
time, where $s_{\lambda}(Y^*)\le k$ is the statistical dimension of $Y^*$, $Y^*$ is an optimal $Y$,
$\eps$ is an error parameter, and $\nnz(A)$ is the number of nonzero  entries of $A$. This is faster than prior work, even when $\lambda=0$.
We also study regularization in a much more general setting. For example, we obtain
sketching-based algorithms for the low-rank approximation problem
$\min_{X,Y} \norm{YX - A}_F^2 + f(Y,X)$
where $f(\cdot,\cdot)$ is a regularizing function satisfying some very general
conditions (chiefly, invariance under orthogonal transformations).

\end{abstract}

\thispagestyle{empty}
\newpage
\setcounter{page}{1}

\section{Introduction}

The technique of matrix sketching, such as the use of random projections,
has been shown in recent years to be a powerful tool for accelerating many important
statistical learning techniques. Indeed, recent work has proposed highly efficient algorithms
for, among other problems, linear regression, low-rank approximation~\cite{Mbook, Wbook} and
canonical correlation analysis~\cite{ABTZ14}. 
In addition to being a powerful theoretical tool, sketching is also an 
applied one; see~\cite{YMM16} for a discussion of state-of-the-art performance for important 
techniques in statistical learning.

Many statistical learning techniques can benefit substantially, in their quality of results, by using some form of regularization.
Regularization can also help by reducing the computing resources needed for these techniques.
%We study the improvements in space and  that regularization can make possible.
While there has been some
prior exploration in this area, as discussed in \S\ref{sub sec results}, commonly it has featured sampling-based techniques, often
focused on regression, and often with analyses using distributional assumptions about the input (though
such assumptions are not always necessary). Our study
considers fast (linear-time) sketching methods, a breadth of problems, and makes no distributional assumptions.
Also, where most prior work studied the distance of an approximate solution to the optimum, our guarantees
are concerning approximation with respect to a relevant loss function - see below for more discussion. 

%
%Research has so far focused largely on using sketching for the ``vanilla'' un-regularized
%versions of the aforementioned problems (we mention some exceptions in \S\ref{sub sec results}). 
%In contrast, it is well appreciated that many statistical learning
%techniques can benefit substantially by using some form of regularization. Thus, there is something of a mismatch between
%the statistical learning techniques that received treatment in the sketching literature, and the statistical techniques that are used on a regular basis.
%
%
%In this paper we study sketching methods for regularized variants of linear regression, low rank approximation and canonical correlation analysis. 
%We give positive answers to the following two questions:
%\begin{enumerate}
%\item Can the gains obtained by sketching be maintained when we consider regularized variants as opposed to 
%  vanilla variants of these problems?
%
%\item Can some forms of regularization actually help in improving the behavior of sketching-based methods for
%  these problems, or in designing more efficient methods?
%\end{enumerate}
%
It is a long-standing theme in the study of randomized algorithms that structures that aid statistical inference can
also aid algorithm design, so that for example, VC dimension and sample compression have been applied
in both areas, and more recently, in cluster analysis the algorithmic advantages of natural statistical assumptions
have been explored. \Ken{not bothering with cites here...}This work is another contribution to this theme. Our
high-level goal in this work 
is to study generic conditions on sketching matrices that can be applied to a wide array of regularized problems
in linear algebra, preserving their objective function values, and exploiting the power of regularization.

\subsection{Results}
\label{sub sec results}
We study regularization both in a fairly broad setting, and in the specific context of the
popular and widely used technique of ridge regularization.
We discuss the latter in sections~\ref{sec ridge regression},~\ref{sec glra} and~\ref{sec:cca};
our main results for ridge regularization,
Theorem~\ref{thm reg stacked}, on linear regression,
Theorem~\ref{thm lowr ridge}, on low-rank approximation,
and
Theorem~\ref{thm approx reg cca}, on canonical correlation analysis,
show that for ridge regularization,
the sketch size need only be a function of the {\em statistical dimension} of the input matrix,
as opposed to its rank, as is common in the analysis of sketching-based methods.
Thus, ridge regularization improves the performance of sketching-based methods.
%For ridge regression, the  \S\ref{subsec large d}, that 

Next, we consider regularizers under rather general assumptions involving
invariance under left and/or right multiplication by orthogonal matrices,
and show that sketching-based methods can be applied, to regularized
multiple-response regression in \S\ref{sec reg mr gen} and to regularized low-rank approximation,
in \S\ref{sec low-rank gen}. Here we obtain running times in terms of the statistical dimension. 
%
%Here the contribution is to show that sketching can be effectively applied;
%we do not have running times in terms of statistical dimension.
Along the way, in \S\ref{subsec low-rank gen svd},
we give a ``base case'' algorithm for reducing low-rank approximation,
via singular value decomposition, to the special case of diagonal matrices.

Throughout we rely on sketching matrix constructions
involving \emph{sparse embeddings} \cite{cw13,nn13,mm13,bdn15,c16},
and on \emph{Sampled Randomized Hadamard Transforms} (SRHT)
\cite{AC06,S06,DMM06a,DMMS07,Tro11,bg12,mdmw12,ldfu13}.
Here for matrix $A$, its sketch is $SA$, where $S$ is a sketching matrix.
The sketching constructions mentioned can be combined to yield a sketching matrix $S$
such that the sketch of matrix $A$, which is simply $SA$, can be computed
in time $O(\nnz(A))$, which is proportional to the number of nonzero entries of $A$.
Moreover, the number of rows of $S$ is small.
Corollary~\ref{cor size of S} summarizes our use of these constructions as applied
to ridge regression.

A key property of a sketching matrix $S$ is that it be a \emph{subspace embedding},
so that $\norm{SAx}_2 \approx\norm{Ax}_2$ for all $x$.
Definition~\ref{def subs embed} gives the technical definition,
and Definition~\ref{def aff emb} gives the definition of the related property
of an \emph{affine embedding} that we
also use. Lemma~\ref{lem AE} summarizes the use of sparse embeddings and
SRHT for subspace and affine embeddings.

%There are sampling matrices
%(another kind of sketching, in the broad sense) with similar properties, and
%these have been applied to ridge regression, but the sampling probabilities
%needed for their construction can be somewhat slower to compute
%and in contrast to the sketches we use, need more than one pass over the data.
%(though this is mostly alleviated by \cite{mdmw12}),
%
\Haim{where do we mention this: Some applications for the problem of \S\ref{sec glra} are given in \cite{cabral2013unifying}.}

In the following we give our main results in more detail.
However, before doing so, we need the formal definition of the statistical dimension.

\begin{definition}[Statistical Dimension]
For real value $\lambda\ge 0$ and rank-$k$
matrix $A$ with singular values $\sigma_i, i\in [k]$, the quantity
$\sd_\lambda(A) \equiv \sum_{i\in [k]}1/(1+\lambda/\sigma_i^2)$
is the \emph{statistical dimension} (or \emph{effective dimension}, or ``hat matrix trace'') of the ridge regression problem with
regularizing weight~$\lambda$.
\end{definition}
Note that $\sd_\lambda(A)$ is decreasing in
$\lambda$, with maximum $\sd_0(A)$ equal to the rank of $A$. Thus a dependence of resources
on $\sd_\lambda(A)$ instead of the rank is never worse, and will be much better for large~$\lambda$.

In \S\ref{sec sd est}, we give an algorithm for estimating $\sd_\lambda(A)$ to within a constant factor,
in $O(\nnz(A))$ time, for $\sd_\lambda(A)\le (n+d)^{1/3}$. Lnowing $\sd_\lambda(A)$
to within a constant factor allows us to set various parameters of our algorithms.

\subsubsection{Ridge Regression}

In \S\ref{sec ridge regression} we apply sketching to reduce from one ridge regression problem to another one with fewer rows.

\begin{theorem}[Less detailed version of Thm.~\ref{thm reg stacked}]
Given $\eps\in (0,1]$ and $A\in\R^{n\times d}$,
there is a sketching distribution over $S\in\R^{m\times n}$,
where $m = \tO(\eps^{-1}\sd_\lambda(A))$,
such that $SA$ can be computed in
\[
O(\nnz(A)) + d \cdot \poly(\sd_\lambda(A)/\eps)
\]
time, and with constant probability
$
\tx\equiv \argmin_{x\in\R^d} \norm{S(Ax-b)}^2 + \lambda\norm{x}^2
$
satisfies
\[
\norm{A\tx-b}^2 + \lambda\norm{\tx}^2 \le (1+\eps)\min_{x\in\R^d} \norm{Ax-b}^2 + \lambda\norm{x}^2.
\]
Here $\poly(\kappa)$ denotes some polynomial function of the value $\kappa$.
\end{theorem}

%While $\sd_\lambda(A)$ is relatively expensive to compute, so that we must assume it unknown here, we can interpret these 
%results as saying that for given computational cost, the output quality improves as $\lambda$ increases.

In our analysis (Lemma~\ref{lem reg}), we map ridge regression to ordinary least squares (by using a matrix with $\sqrt{\lambda}\Iden$
adjoined), and then apply prior analysis of sketching algorithms, but with the novel use of a sketching
matrix that is ``partly exact''; this latter step is important to obtain our overall bounds.
We also show that sketching matrices can be usefully composed 
in our regularized setting; this is straightforward in the non-regularized case, but requires some work here.

As noted, the statistical dimension of a data matrix in the context of ridge regression is also referred
to as the {\em effective degrees of freedom} of the regression problem in the statistics literature,
and the statistical dimension features, as the name suggests, in the statistical analysis of the method.
Our results show that the statistical dimension affects not only the statistical capacity of ridge regression,
but also its computational complexity.

The reduction of the above theorem is mainly of interest when $n\gg \sd_\lambda(A)$,
which holds in particular when $n\gg d$, since $d\ge\rank(A) \ge \sd_\lambda(A)$. We also give a reduction using sketching
when $d$ is large, discussed in \S\ref{subsec large d}. Here algorithmic resources depend on a power of
$\sigma_1^2/\lambda$, where $\sigma_1$ is the leading singular value of $A$. This result falls within our theme of
improved efficiency as $\lambda$ increases, but in contrast to our other results, performance does not degrade gracefully
as $\lambda\rightarrow 0$. The difficulty is that we use the product of sketches $A S^\top S A^\top$ to estimate the product
$AA^\top$ in the expression $\norm{AA^\top y - b}$. Since that expression can be zero, and since we
seek a strong notion of relative error, the error of our overall estimate
is harder to control, and impossible when $\lambda=0$.

As for related work on ridge regression, Lu \emph{et al.} \cite{LDFU} apply the SRHT to ridge regression,
analyzing the statistical risk under the distributional assumption on the input data that
$b$ is a random variable, and not giving bounds in terms of $\sd_\lambda$.
El Alaoui \emph{et al.} \cite{AM14} apply sampling techniques based on the \emph{leverage scores} of a
matrix derived from the input,
with a different error measure than ours, namely, the statistical risk; here for their error analysis
they consider the case when the noise in their ridge regression
problem is i.i.d. Gaussian. They give results in terms of $\sd_\lambda(A)$, which arises 
naturally for them as the sum of the leverage scores. Here we show that
this quantity arises also in the context of oblivious subspace embeddings, and with the goal being to
obtain a worst-case relative-error guarantee in objective function value rather than for minimizing
statistical risk. 
Chen \emph{et al.} \cite{CLLKZ} apply sparse embeddings to ridge regression, obtaining
solutions $\tx$ with $\norm{\tx-x^*}_2$ small, where $x^*$ is optimal,
and do this in $O(\nnz(A) + d^3/\eps^2)$ time. They also analyze the statistical risk of their
output. Yang \emph{et al.} \cite{YPW} consider slower sketching methods than those here, and analyze their error
under distributional assumptions using an incomparable notion of statistical dimension.  Frostig \emph{et al.} \cite{FGKS}
make distributional assumptions, in particular a kurtosis property. Frostig \emph{et al.} \cite{FGKS_ICML}
give bounds in terms of a convex condition number that can be much larger than $\sd_\lambda(A)$.
In~\cite{ACW16} we analyze using random features to form preconditioners for use
in kernel ridge regression. We show that the number of random features required for an high quality preconditioner
is a function of the statistical dimensions, much like the results in this paper.
Another related work is that of Pilanci \emph{et al.} \cite{PW14} which we dicuss below.

\subsubsection{Ridge Low-rank Approximation}

In  \S\ref{sec glra} we consider the following problem:
for given $A\in\R^{n\times d}$, integer $k$, and weight $\lambda\ge 0$,
find:
\begin{equation}\label{eq lowr first}
\min_{\substack{Y\in\R^{n\times k}\\ X\in\R^{k\times d}}}
		\norm{YX - A}_F^2 + \lambda\norm{Y}_F^2 + \lambda\norm{X}_F^2,
\end{equation}
where, as is well known (and discussed in detail later), this regularization term is equivalent to $2\lambda\norm{YX}_*$, 
where $\nm{*}$ is the trace (nuclear) norm, the Schatten 1-norm. We show the following.

\begin{theorem}[Less detailed Thm.~\ref{thm lowr ridge}]
Given input $A\in\R^{n\times d}$, there is a  sketching-based algorithm
returning $\tY\in\R^{n\times k}, \tX\in\R^{k\times d}$ such that
with constant probability,
$\tY$ and $\tX$ form a $(1+\eps)$-approximate minimizer to \eqref{eq lowr first}, that is,
\begin{align}
\norm{\tY & \tX - A}_F^2 +  \lambda\norm{\tY}_F^2 + \lambda\norm{\tX}_F^2
\\ & \le (1+\eps) \min_{\substack{Y\in\R^{n\times k}\\ X\in\R^{k\times d}}}
		\norm{YX - A}_F^2 + \lambda\norm{Y}_F^2 + \lambda\norm{X}_F^2.
\end{align}
The matrices $\tY$ and $\tX$ can be found in
$
O(\nnz(A)) + \tO((n+d)\eps^{-1}k \min\{k, \eps^{-1}\sd_\lambda(Y^*)\})+ \poly(\eps^{-1} \sd_\lambda(Y^*))
$
%\tO(\eps^{-8} \sd_\lambda(Y^*)^3)$
time, where $Y^*$  is an optimum $Y$ in \eqref{eq lowr first} such that
$\sd_\lambda(X^*)=\sd_\lambda(Y^*)\le \rank(Y^*)\le k$.
\end{theorem}

This algorithm follows other algorithms for $\lambda=0$
with running times of the form
$
O(\nnz(A)) + (n+d)\poly(k/\eps)$ 
(e.g. \cite{cw13}), and has the best known dependence on
$k$ and $\eps$ for algorithms of this type, even when $\lambda=0$.

\Ken{Methods}
Our approach is to first extend our ridge regression results
to the multiple-response case $\min_Z \norm{AZ-B}_F^2 + \lambda\norm{Z}_F^2$,
and then reduce the multiple-response problem to a smaller one by
showing that up to a cost in solution quality, we can assume that each row of $Z$ lies
in the rowspace of $SA$, for $S$ a suitable sketching matrix. We apply this observation
twice to the low-rank approximation problem, so that $Y$ can be assumed to
be of the form $AR\tY$, and $X$ of the form $\tX SA$, for sketching matrix $S$ and (right) sketching
matrix $R$. Another round of sketching
then reduces to a low-rank approximation problem of size independent of $n$ and $d$,
and finally an SVD-based method is applied to that small problem.

Regarding related work: the regularization ``encourages'' the rank of $YX$ to be small, even when there
is no rank constraint ($k$ is large),
and this unconstrained problem has been extensively studied;
even so, the rank constraint can reduce the computational cost
and improve the output quality, as discussed by \cite{cabral2013unifying},
who also give further background, and who give experimental results
on an iterative algorithm.
Pilanci \emph{et al.} \cite{PW14} consider only algorithms where the sketching time is at least $\Omega(nd)$,
which can be much slower than our $\nnz(A)$ for sparse matrices, and it is not clear if their techniques can be extended.
In the case of low-rank approximation with a nuclear norm constraint (the closest to our work), as the authors note,
their paper gives no improvement in running time. While their framework might imply analyses for ridge regression,
they did not consider it specifically, and such an analysis may not follow directly.

\subsubsection{Regularized Canonical Correlation Analysis}
Canonical correlation analysis (CCA) is an important statistical technique whose input
is a pair of matrices, and whose solution depends on the Gram matrices $A^\top A$ and $B^\top B$.
If these Gram matrices are ill-conditioned it is useful to regularize them by instead
using $A^\top A+\lambda_1\Iden_d$
and $B^\top B+ \lambda_2\Iden_{d'}$, for weights $\lambda_1, \lambda_2 \ge 0$.
Thus, in this paper we consider a regularized version of CCA, defined as follows 
(our definition is in the same spirit as the one used by~\cite{ABTZ14}).
\begin{definition}
Let $A\in\R^{n\times d}$ and $B\in\R^{n\times {d'}}$, and let
$$q=\min (\rank(A^\top A + \lambda_1 \Iden_d), \rank(B^\top B + \lambda_2 \Iden_{d'})).$$
Let $\lambda_1 \geq 0$ and $\lambda_2 \geq 0$. The {\em $(\lambda_1, \lambda_2)$ canonical correlations }
$\sigma^{(\lambda_1, \lambda_2)}_1 \geq \dots \geq \sigma^{(\lambda_1, \lambda_2)}_q$ and
{\em $(\lambda_1, \lambda_2)$ canonical weights } $u_1,\dots,u_q\in \R^d$ and $v_1,\dots,v_q\in \R^{d'}$
are ones that maximize
$$
\tr(U^\top A^\top B V)
$$
subject to
\begin{eqnarray*}
  U^\top (A^\top A + \lambda_1 \Iden_d) U &=& I_q \\
  V^\top (B^\top B + \lambda_2 \Iden_{d'}) V &=& I_q \\
  U^\top A^\top B V &=& \diag(\sigma^{(\lambda_1, \lambda_2)}_1,\dots ,\sigma^{(\lambda_1, \lambda_2)}_q)
\end{eqnarray*}
where $U=\left[u_1,\dots,u_q \right]\in\R^{n\times q}$ and $V=\left[v_1,\dots,v_q \right]\in\R^{d'\times q}$.
\end{definition}

One classical way to solve non-regularized CCA ($\lambda_1 = \lambda_2 = 0$) is the Bj\"{o}rck-Golub algorithm~\cite{BG73}. In \S\ref{sec:cca} we show that regularized CCA can be solved using a variant of the Bj\"{o}rck-Golub algorithm.

Avron et al.~\cite{ABTZ14} showed how to use sketching to compute an approximate CCA.
In \S\ref{sec:cca} we show how to use sketching to compute an approximate regularized CCA.
\begin{theorem}[Loose version of Thm.~\ref{thm approx reg cca}]
  \Haim{The sketch size can be improved using composition, but probably too late
    in the game to do it}
There is a distribution over matrices $S\in\R^{m\times n}$ with
$m=O(\max(\sd_{\lambda_1}(A), \sd_{\lambda_2}(B))^2 / \epsilon^2)$
such that with constant probability,
the regularized CCA of $(SA, SB)$ is an $\epsilon$-approximate CCA of $(A,B)$.
The matrices $SA$ and $SB$ can be computed in $O(\nnz(A)+\nnz(B))$ time.
\end{theorem}

Our generalization of the classical Bj\"{o}rck-Golub algorithm shows that
regularized canonical correlation analysis can be computed via the product of two
matrices whose columns are non-orthogonal regularized bases of $A$ and $B$.
We then show that these two matrices are easier to sketch than the orthogonal
bases that arise in non-regularized CCA. This in turn can be tied to approximation
bounds of sketched regularized CCA versus exact CCA.

\subsubsection{General Regularization}

A key property of the Frobenius norm $\nm{F}$ is that it is invariant under
rotations; for example, it satisfies the \emph{right orthogonal
invariance} condition $\norm{AQ}_F = \norm{A}_F$, for any orthogonal matrix $Q$
(assuming, of course, that $A$ and $Q$ having dimensions so that $AQ$ is defined).
In  \S\ref{sec reg mr gen} and \S\ref{sec low-rank gen},
we study conditions under which such an invariance property, and little else, is enough
to allow fast sketching-based approximation algorithms.

For regularized multiple-response regression, we have the following.
\begin{theorem}[Implied by Thm.~\ref{thm gen regul regr}]
Let $f(\cdot)$ be a real-valued function on matrices that is right orthogonally invariant, subadditive, and
invariant under padding the input matrix by rows or columns of zeros.
Let $A\in\R^{n\times d}, B\in\R^{n\times {d'}}$. 
%Let $X^*$ and $\Delta^*$ as in Lemma~\ref{lem rowspace SB}.
Suppose that for $r\equiv\rank A$,
there is an algorithm that for general $n,d,{d'},r$ and $\eps > 0$,
in time $\tau(d,n,{d'},r,\eps)$
finds $\tX$ with
\[
\norm{A\tX-B}_F^2 + f(\tX) \le (1+\eps)\min_{X\in\R^{d\times d'}} \norm{AX-B}_F^2 + f(X).
\]
Then there is another algorithm that with constant
probability finds such an $\tX$, taking time
\[
O(\nnz(A) + \nnz(B) + (n+d+{d'})\poly(r/\eps)) + \tau(d, \poly(r/\eps),\poly(r/\eps), r,\eps).
\]
\end{theorem}
(Note that Thm.~\ref{thm gen regul regr} seemingly requires an additional property called {\em sketching inheritance}.
However this condition is implied by the conditions of the last theorem.)

That is, sketching can be used to reduce to a problem in which the only remaining large matrix dimension
is $d$, the number of columns of $A$.

This reduction is a building block for our results for regularized low-rank approximation.
Here the regularizer is a real-valued function $f(Y,X)$ on matrices $Y\in\R^{n\times k}, X\in\R^{k\times d}$.
We show that under broad conditions on $f(\cdot, \cdot)$, sketching can be applied to
\begin{equation}\label{eq lowr gen first}
\min_{\substack{Y\in\R^{n\times k}\\ X\in\R^{k\times d}}}
		\norm{YX - A}_F^2 + f(Y,X).
\end{equation}
Our conditions imply fast algorithms when, for example, $f(Y,X)=\norm{YX}_{(p)}$,
where $\nm{(p)}$ is a Schatten $p$-norm,
or when $f(Y,X) = \min\{\lambda_1 \norm{YX}_{(1)}, \lambda_2 \norm{YX}_{(2)} \}$,
for weights $\lambda_1, \lambda_2$, and more. Of course, there are norms, such as the entriwise $\ell_1$ norm,
that do not satisfy these orthogonal invariance conditions.

\begin{theorem}[Implied by Thm.~\ref{thm lowr gen sk}]
Let $f(Y,X)$ be a real-valued function on matrices that in each argument is subadditive and
invariant under padding by rows or columns of zeros, and also
right orthogonally invariant in its right argument and left orthogonally invariant in its left argument.

Suppose there is a procedure that solves \eqref{eq lowr gen first} when $A$, $Y$, and $X$ are $k\times k$ matrices,
and $A$ is diagonal, and $YX$ is constrained to be diagonal, taking time
$\tau(k)$ for a function $\tau(\cdot)$.

%Let $(Y^*, X^*)$ be an optimal solution pair:
%\begin{equation}\label{eq lowr gener}
%Y^*,X^*\equiv \argmin_{\substack{Y\in\R^{n\times k}\\ X\in\R^{k\times d}}} \norm{YX-A}_F^2 + f(Y,X),
%\end{equation}
%and $\Delta^* \equiv \norm{Y^*X^*-A}_F^2 + f(Y^*, X^*)$.

Then for general $A$, there is an algorithm that finds a $(1+\eps)$-approximate solution $(\tY, \tX)$ in time
$
O(\nnz(A)) + \tO(n+d)\poly(k/\eps)  + \tau(k).$ 
\end{theorem}

The proof involves a reduction to small matrices, followed by
a reduction, discussed in \S\ref{subsec low-rank gen svd}, that uses the SVD to reduce to the diagonal case.
This result, Corollary~\ref{cor simple sol}, generalizes results of  \cite{UHZB}, who gave such a reduction for
$f(Y,X)=\norm{X}_F^2 + \norm{Y}_F^2$; also, we give a very different proof.

As for related work, \cite{UHZB} survey and extend work in this setting, and propose iterative algorithms for
this problem. The regularizers $f(Y,X)$ they consider, and evaluate experimentally,
are more general than we can analyze.

The conditions on $f(Y,X)$ are quite general; it may be that for some instances, the resulting problem is NP-hard.
Here our reduction would be especially interesting, because the size of the reduced NP-hard problem depends only
on $k$.

%\section{Preliminaries}

\subsection{Basic Definitions and Notation}

%\Ken{if not elsewhere}
We denote scalars using Greek letters. Vectors
are denoted by $x,y,\dots$ and matrices by $A,B,\dots$.
We use the convention that vectors are column-vectors. We use $\nnz(\cdot)$ to
denote the number of nonzeros in a vector or matrix. We denote by $[n]$ the
set ${1,\dots,n}$. The notation $\alpha = (1 \pm \gamma)\beta$ means
that $(1- \gamma)\beta \leq \alpha \leq (1 + \gamma)\beta$.

Throughout the paper, $A$ denotes an $n \times d$ matrix, and $\sigma_1 \geq \sigma_2 \ge \cdots \ge \sigma_{\min(n,d)}$
its singular values.

\begin{definition}[Schatten $p$-norm]
The \emph{Schatten $p$-norm}  of $A$ is $\norm{A}_{(p)}\equiv \left[\sum_i \sigma_i^p\right]^{1/p}$.
Note that the trace (nuclear) norm $\norm{A}_*=\norm{A}_{(1)}$,
the Frobenius norm $\norm{A}_F = \norm{A}_{(2)}$,
and the spectral norm $\norm{A}_2 = \norm{A}_{(\infty)}$.
\end{definition}

The notation $\norm{\cdot}$
without a subscript denotes the $\ell_2$ norm for vectors, and the spectral norm
for matrices. We use a subscript for other norms.
We use $\range(A)$ to denote the subspace spanned by the columns of $A$, i.e.
$\range(A) \equiv \{Ax\mid x\in\R^d\}$. $\Iden_d$ denotes the $d\times d$ identity matrix,
$0_d$ denotes the column vector comprising $d$ entries of zero, and $0_{a\times b}\in\R^{a\times b}$
denotes a zero matrix.

The rank $\rank(A)$ of a matrix $A$ is the dimension of the subspace $\range(A)$
spanned by its columns (equivalently,
the number of its non-zero singular values). Bounds on sketch sizes are often written in terms
of the rank of the matrices involved.
\Haim{I think a sentence that explains why we give the definitions below, will not be bad...}
%In this paper, we show that for regularized problems
%you often can have bounds in terms of more refined notions of rank that take into consideration
%the regularization parameter and the distribution of singular values. We give definitions below.

\begin{definition}[Stable Rank]
The \emph{stable rank} $\sr(A) \equiv \norm{A}_F^2/\norm{A}_2^2$. The stable rank satisfies $\sr(A)\le \rank(A)$.
\end{definition}

%
%\subsection{Related Work}
%\label{sec:related}
%
%
%
%\Ken{
%Pilanci/Wainwright http://arxiv.org/abs/1404.7203 :
%dense transforms only; very large sketch sizes for low-rank approximation; possibly interesting for $L_1$-constrained least squares; shows that sketch size depends only on rank(A) and not d, which is immediate for the OSEs they discuss}
%
%\Ken{
%Pilanci/Wainwright http://arxiv.org/abs/1404.7203:
%(Randomized sketches for kernels: Fast and optimal non-parametric regression)\\
%Kernel regression; dense transforms}
%
%\Ken{
%El Alaoui / Mahoney http://arxiv.org/abs/1411.0306:\\
%Fast Randomized Kernel Methods With Statistical Guarantees\\
%Sampling using "ridge" version of leverage scores\\
%any overlap with "structural" result, Theorem 1?
%}

%\input{sketch_right}
%\noindent {\bf Paper Outline:} Due to space constraints, all proofs are omitted,
%and all results except our results for ridge regression are deferred to the appendix.

% !TEX root = p.tex

\section{Ridge Regression}
\label{sec ridge regression}

Let $A\in\R^{n\times d}$, $b\in\R^n$, and $\lambda>0$. 
In this section we consider the \emph{ridge regression} problem:
\begin{equation}\label{eq ridge1}
\min_{x\in\R^d} \norm{Ax-b}^2 + \lambda\norm{x}^2,
\end{equation}
Let
\begin{align*}
x^* & \equiv \argmin_{x\in\R^d} \norm{Ax-b}^2 + \lambda\norm{x}^2\text{\ and}
\\ \Delta_* & \equiv  \norm{Ax^*-b}^2 + \lambda\norm{x^*}^2.
\end{align*}
In general $x^* = (A^\top A + \lambda I_d)^{-1} A^\top b = A^\top (A A^\top + \lambda I_n)^{-1}b$,
so $x^\star$ can be found in $O(\nnz(A) \min(n, d))$ time using an iterative method (e.g., LSQR). 
Our goal in this section is to design faster algorithms that find an approximate
$\tx$ in the following sense:
\begin{equation}
\label{eq:goal ridge regression}
\norm{A\tx-b}^2 + \lambda\norm{\tx}^2\le (1+\eps)\Delta_*\,.
\end{equation}
In our analysis, we distinguish between two cases: $n \gg d$ and $d \gg n$.

\begin{remark}
In this paper we consider only approximations of the form~\eqref{eq:goal ridge regression}.
Although we do not explore it in this paper, our techniques can also be used to derive
preconditioned methods. Analysis of preconditioned kernel ridge regression, which is related
to the $d \gg n$ case, is explored in~\cite{ACW16}.
 
\end{remark}
\subsection{Large $n$}
\label{sec large n}
In this subsection we design an algorithm that is aimed at the case when 
$n \gg d$. However, the results themselves are correct even when $n < d$. The general strategy
is to design a distribution on matrices of size $m$-by-$n$ ($m$ is a parameter), 
sample an $S$ from that distribution, and solve 
$\tx \equiv \argmin_{x\in\R^d} \norm{S(Ax-b)}^2 +  \lambda\norm{x}^2\,.$

The following lemma defines conditions on the distribution that guarantee 
that~\eqref{eq:goal ridge regression} holds with constant probability
(which can be boosted to high probability by repetition and taking the solution with
minimum objective value). \Haim{might be obvious to COLT readers, but still not bad to mention it, no?}
\begin{lemma}\label{lem reg}
Let $x^*\in\R^d$, $A$ and $b$ as above.
Let $U_1\in\R^{n\times d}$ comprise the first $n$ rows of an orthogonal basis
for $\twomat{A}{\sqrt{\lambda} \Iden_d}$.
Let sketching matrix $S\in\R^{m\times n}$ have a distribution such that with constant probability
\begin{equation}\label{eq prod U1}
\norm{U_1^\top S^\top S U_1 - U_1^\top U_1}_2\le 1/4,
\end{equation}
and
\begin{equation}\label{eq prod}
\norm{U_1^\top S^\top S(b-Ax^*) - U_1^\top (b-Ax^*)} \le \sqrt{\eps \Delta_*/2}.
\end{equation}
Then
with constant probability,
$\tx \equiv \argmin_{x\in\R^d} \norm{S(Ax-b)}^2 +  \lambda\norm{x}^2$
has
$$\norm{A\tx-b}^2 + \lambda\norm{\tx}^2\le (1+\eps)\Delta_*.$$
\end{lemma}

\def\lemregproof{
\begin{proof}
Let $\hat{A}\in\R^{(n+d)\times d}$ have orthonormal columns
with $\range(\hat{A} )=\range( \twomat{A}{\sqrt{\lambda} \Iden_d} )$.
(An explicit expression for one such $\hat{A}$ is given below.)
Let $\hat{b} \equiv \twomat{b}{0_d}$.
We have
\begin{equation}\label{eq ridge2}
\min_{y\in\R^d} \norm{\hat{A}y - \hat{b}}_2
\end{equation}
equivalent to \eqref{eq ridge1}, in the sense that for any $\hat{A}y\in\range(\hat{A})$, there is $x\in\R^d$ with
$\hat{A}y = \twomat{A}{\sqrt{\lambda} \Iden_d} x$,
so that
$\norm{\hA y - \hat{b}}^2 = \norm{\twomat{A}{\sqrt{\lambda} \Iden_d} x - \hat{b}}^2
=\norm{b-Ax}^2 +  \lambda\norm{x}^2$.
Let $y^* = \argmin_{y\in\R^d} \norm{\hat{A}y - \hat{b}}_2$,
so that $\hA y^* = \twomat{Ax^*}{\sqrt{\lambda} x^*}$.

Let $\hat{A} =  \twomat{U_1}{U_2}$, where $U_1\in\R^{n\times d}$ and
$U_2\in\R^{d\times d}$, so that $U_1$ is as in the lemma statement.

Let $\hS \equiv \left[\begin{smallmatrix} S & 0_{m\times d}\\ 0_{d\times n} & \Iden_d \end{smallmatrix}\right]$.

Using \eqref{eq prod U1}, with constant probability
\begin{equation}\label{eq embed2}
\norm{\hA^\top \hS^\top \hS \hA - \Iden_d}_2
	= \norm{U_1^\top S^\top S U_1 + U_2^\top U_2 - \Iden_d}_2
	= \norm{U_1^\top S^\top S U_1 - U_1^\top U_1}_2
	\le 1/4.
\end{equation}
Using the normal equations for \eqref{eq ridge2}, we have
\begin{equation*}
0
	= \hA^\top (\hat{b} - \hA y^*)
	= U_1^\top (b-Ax^*) - \sqrt{\lambda} U_2^\top x^*,
\end{equation*}
and so
\[
\hA^\top \hS^\top \hS (\hat{b} - \hA y^*)
	= U_1^\top S^\top S (b-Ax^*) - \sqrt{\lambda} U_2^\top x^*
	= U_1^\top S^\top S (b-Ax^*)  - U_1^\top (b-Ax^*).
\]
Using \eqref{eq prod}, with constant probability
\begin{align}
\norm{\hA^\top \hS^\top \hS (\hat{b} - \hA y^*)}
	   & = \norm{U_1^\top S^\top S (b-Ax^*)  - U_1^\top (b-Ax^*)} \nonumber
	\\ & \le \sqrt{\eps \Delta_*/2} = \sqrt{\eps/2} \norm{\hat{b} - \hA y^*}. \label{eq prod2}
\end{align}
It follows by a standard result \Ken{who exactly?} from \eqref{eq embed2} and \eqref{eq prod2} that the solution
$\tilde{y} \equiv \argmin_{y\in\R^d} \norm{\hS(\hA y - \hat{b})}$ has
$\norm{\hA\tilde{y} - \hat{b}}\le (1+\eps)\min_{y\in\R^d} \norm{\hat{A}y - \hat{b}}$,
and therefore that $\tilde{x}$ satisfies the claim of the theorem.

For convenience we give the proof of the standard result: \eqref{eq embed2} implies that
$\hA^\top \hS^\top \hS\hA$ has smallest singular value at least $3/4$. The normal equations for the unsketched
and sketched problems are
$$\hA^\top(\hat{b} - \hA y^*) = 0 = \hA^\top \hS^\top \hS (\hat{b} - \hA \tilde{y}).$$
The normal equations for the unsketched case imply
$\norm{\hA \ty - \hb}^2 = \norm{\hA(\ty - y^*)}^2 + \norm{\hb - \hA y^*}^2$,
so it is enough to show that $\norm{\hA(\ty - y^*)}^2 = \norm{\ty - y^*}^2 \le \eps\Delta_*$.
We have
\begin{align*}
(3/4) \norm{\ty - y^*}
	   & \le \norm{\hA^\top \hS^\top \hS \hA(\ty - y^*)} & \text{using \eqref{eq embed2}}
	\\ & = \norm{\hA^\top \hS^\top \hS \hA(\ty - y^*) - \hA^\top \hS^\top \hS (\hat{b} - \hA \tilde{y})} &  \text{normal eqs}
	\\ & = \norm{\hA^\top \hS^\top \hS(\hb - \hA y^*)} &
	\\ & \le \sqrt{\eps\Delta_*/2}  & \text{from \eqref{eq prod2}},
\end{align*}
so that $\norm{\ty - y^*}^2 \le (4/3)^2\eps\Delta_*/2\le \eps\Delta_*$. The theorem follows.
\end{proof}
}
\ifSUB
\begin{proof}
Please see \S\ref{subsec lemregproof}.
\end{proof}
\else
\lemregproof
\fi

\begin{lemma}\label{lem U1 size}
For $U_1$ as in Lemma~\ref{lem reg}, $\norm{U_1}_F^2 = \sd_\lambda(A)  =  \sum_i 1/(1+ \lambda/\sigma_i^2)$, where
$A$ has singular values $\sigma_i$. Also $\norm{U_1}_2 = 1/\sqrt{1+ \lambda/\sigma_1^2}$.
\end{lemma}

This follows from (3.47) of \cite{HTF}; for completeness, a proof is given here.

%@book{HTF,
%title={The Elements of Statistical Learning},
%author={Trevor Hastie and Robert Tibshirani and Jerome Friedman},
%year=2013}

\def\lemUonesizeproof{
\begin{proof}
Suppose $A=U\Sigma V^\top$, the full SVD, so that $U\in\R^{n\times n}$, $\Sigma\in\R^{n\times d}$, and $V\in\R^{d\times d}$.
Let $D\equiv (\Sigma^\top \Sigma + \lambda\Iden_d)^{-1/2}$. Then
$\hA = \twomat{U\Sigma D}{V \sqrt{\lambda} D}$
has $\hA^\top \hA = \Iden_d$, and for given $x$, there is $y=D^{-1} V^\top x$ with
$\hA y = \twomat{A}{\sqrt{\lambda} \Iden_d} x$.
We have $\norm{U_1}_F^2 = \norm{U\Sigma D}_F^2 = \norm{\Sigma D}_F^2 =  \sum_i 1/(1+ \lambda/\sigma_i^2)$
as claimed. Also $\norm{U_1}_2 = \norm{U \Sigma D}_2 = \norm{\Sigma D}_2 = 1/\sqrt{1+ \lambda/\sigma_1^2}$,
and the lemma follows.
\end{proof}
} %lemUonesizeproof
\ifSUB
\begin{proof}
Please see \S\ref{subsec lemUonesizeproof}.
\end{proof}
\else
\lemUonesizeproof
\fi

\begin{definition}[large $\lambda$]\label{def large lam}
Say that $\lambda$ is \emph{large} for $A$ with largest singular value
$\sigma_1$, and error parameter $\eps$, if $\lambda/\sigma_1^2\ge 1/\eps$.
\end{definition}

The following lemma implies that if $\lambda$ is large, then
$x=0$ is a good approximate solution, and so long as we include a check that a proposed
solution is no worse than $x=0$, we can assume that $\lambda$ is not large.

\begin{lemma}\label{lem lam large}
For $\eps\in (0,1]$, large $\lambda$, and all $x$,
$\norm{Ax-b}^2 + \lambda\norm{x}^2 \ge \norm{b}^2/(1+\eps)$.
If $\lambda$ is not large then $\norm{U_1}_2^2 \ge \eps/2$.
\end{lemma}

\def\lemlamlargeproof{
\begin{proof}
If $\sigma_1\norm{x}\ge\norm{b}$, then $\lambda\norm{x}^2\ge \sigma_1^2\norm{x}^2\ge\norm{b}^2$.
Suppose $\sigma_1\norm{x}\le\norm{b}$. Then:
\begin{align*}
\norm{Ax-b}^2 + \lambda\norm{x}^2
	   &  = \norm{Ax}^2 + \norm{b}^2 - 2b^\top Ax + \lambda\norm{x}^2
	\\ & \ge (\norm{b}- \norm{Ax} )^2 + \lambda\norm{x}^2 & \text{Cauchy-Schwartz}
	\\ & \ge (\norm{b} - \sigma_1\norm{x})^2 + \lambda\norm{x}^2 & \text{assumption}
	\\ & \ge \norm{b}^2 /(1+\sigma_1^2/\lambda) & \text{calculus}
	\\ & \ge \norm{b}^2/(1+\eps), &\text{large $\lambda$}
\end{align*}
as claimed. The last statement follows from Lemma~\ref{lem U1 size}.
\end{proof}
}

\ifSUB
\begin{proof}
Please see \S\ref{subsec lemlamlargeproof}.
\end{proof}
\else
\lemlamlargeproof
\fi
Below we discuss possibilities for choosing the sketching matrix $S$. 
%One possibility for the sketching matrix $S$ above
%is obtained by composing a sparse embedding matrix with a single
%non-zero entry per column \cite{cw13,mm13,nn13} with 
%an SRHT; this has a leading order term of $\nnz(A)$ with a very small constant.
%We focus on this composition below, but we stress that
%other possibilities based on composing an OSNAP \cite{nn13,bdn15,c16} with an SRHT or a dense
%subgaussian matrix may be possible.
We want to emphasize
that the first condition in Lemma \ref{lem reg} is {\it not} a subspace embedding guarantee,
despite having superficial similarity. 
Indeed, notice that the columns of $U_1$ are not orthonormal, since we only take the
first $n$ rows of an orthogonal basis of $\twomat{A}{\sqrt{\lambda} \Iden_d}$. Rather,
the first condition is an instance of approximate matrix product with a spectral norm
guarantee with constant error,
for which optimal bounds in terms of the stable rank $\sr(U_1)$
were recently obtained \cite{CNW}. As we discuss in the proof of part (i)
of Corollary \ref{cor size of S} below, $\sr(U_1)$ is upper bounded by 
$\sd_{\lambda}(A)/\epsilon$. 

We only mention a few possibilities of sketching matrix $S$ below, though others
are possible with different tradeoffs and compositions. 
\begin{corollary}\label{cor size of S}
Suppose $\lambda$ is not large (Def.~\ref{def large lam}).
 There is a constant $K>0$ such that for
\begin{enumerate}[i.]
\item $m \ge K(\eps^{-1} \sd_\lambda(A) + \sd_\lambda(A)^2)$ and $S\in\R^{m\times n}$ a sparse embedding matrix (see \cite{cw13,mm13,nn13}) with $SA$
  computable in $O(\nnz(A))$ time, or one can choose $m \ge K (\eps^{-1} \sd_\lambda(A) + \min((\sd_\lambda(A)/\epsilon)^{1+\gamma}, \sd_\lambda(A)^2))$
  an OSNAP (see \cite{nn13,bdn15,c16}) with $SA$ computable in $O(\nnz(A))$ time,
  where $\gamma > 0$ is an arbitrarily small constant, or 
  \item $m\ge K\eps^{-1} (\sd_\lambda(A) + \log(1/\eps))\log(\sd_\lambda(A)/\eps)$ and $S\in\R^{m\times n}$ a Subsampled Randomized Hadamard Transform (SRHT) embedding matrix (see, e.g., \cite{bg12}), with $SA$ computable in $O(nd\log n)$ time, or
\item $m\ge K\eps^{-1}\sd_\lambda(A)$ and $S\in\R^{m\times n}$ a matrix of i.i.d. subgaussian values
with $SA$ computable in $O(ndm)$ time,
\end{enumerate}
the conditions \eqref{eq prod U1} and \eqref{eq prod} of Lemma~\ref{lem reg} apply, and with constant probability
the corresponding $\tilde{x} = \argmin_{x\in\R^d} \norm{S(Ax-b)} + \lambda\norm{x}^2$
is an $\eps$-approximate solution to $\min_{x\in\R^d} \norm{b-Ax}^2 + \lambda\norm{x}^2$.
\end{corollary}
 
\def\corsizeofSproof{
\begin{proof}
Recall that $\sd_\lambda(A) = \norm{U_1}_F^2$.
For (i): sparse embedding distributions satisfy the bound for matrix multiplication
$$\norm{W^\top S^\top S H - W^\top H}_F \le C\norm{W}_F\norm{H}_F/\sqrt{m},$$
for a constant $C$ \cite{cw13,mm13,nn13}; this is also true of OSNAP matrices.
We set $W=H=U_1$ and use $\norm{X}_2 \le \norm{X}_F$ for all $X$ and
$m\ge K \norm{U_1}_F^4$ to obtain \eqref{eq prod U1}, and set $W=U_1$, $H=b - Ax^*$
and use $m\ge K \norm{U_1}_F^2/\eps$ to obtain \eqref{eq prod}. (Here the bound is slightly stronger
than \eqref{eq prod}, holding for $\lambda=0$.) With \eqref{eq prod U1} and \eqref{eq prod}, the
claim for $\tilde{x}$ from a sparse embedding follows using Lemma~\ref{lem reg}.

For OSNAP, Theorem 1 in \cite{CNW} together with \cite{nn13}
imply that for $m = O(\sr(U_1)^{1+\gamma})$, condition
\eqref{eq prod U1} holds. Here $\sr(U_1) = \frac{\|U_1\|_F^2}{\|U_1\|_2^2}$, and by 
Lemma \ref{lem U1 size} and Lemma \ref{lem lam large}, $\sr(U_1) \leq \sd_{\lambda}(A)/\epsilon$.
We note that \eqref{eq prod} continues to hold as in the previous paragraph. Thus, $m$ is at most the min of $O((\sd_{\lambda}(A)/\epsilon)^{1+\gamma})$ and $O(\sd_{\lambda}(A)/\epsilon + \sd_{\lambda}(A)^2)$. 

For (ii): Theorems~1 and~9 of \cite{CNW} imply that for $\gamma\le 1$, with constant probability
\begin{equation}\label{eq cnw}
\norm{W^\top S^\top S H - W^\top H}_2 \le \gamma\norm{W}_2\norm{H}_2
\end{equation}
for SRHT $S$, when $$m\ge C(\sr(W) + \sr(H) + \log(1/\gamma))\log(\sr(W) + \sr(H) )/\gamma^2$$ for a constant $C$.
We let $W=H=U_1$ and $\gamma=\min\{1, 1/4\norm{U_1}^2\}$. We have
\[
\norm{U_1^\top S^\top SU_1 - U_1^\top U_1}_2
	\le \min\{1, 1/4\norm{U_1}^2\} \norm{U_1}_2^2 = \min\{\norm{U_1}_2^2, 1/4\} \le 1/4,
\]
and
\[
\sr(U_1)/\gamma^2
	= \frac{\norm{U_1}_F^2}{\norm{U_1}_2^2}\max\{1, 4\norm{U_1}_2^2\}
	= \norm{U_1}_F^2 \max\{1/\norm{U_1}_2^2, 4\} \le 2\norm{U_1}_F^2/\eps
\]
using Lemma~\ref{lem lam large} and the assumption that $\lambda$ is large.
(And assuming $\eps\le 1/2$.)
Noting that $\log(1/\gamma) = O(\log(1/\eps))$ 
and $\log(\sr(U_1))=O(\log \norm{U_1}_F/\eps)$ using Lemma~\ref{lem lam large},
we have that $m$ as claimed suffices for \eqref{eq prod U1}. 

For \eqref{eq prod}, we use \eqref{eq cnw} with $W=U_1$, $H=Ax^*-b$,
and $\gamma= \sqrt{\eps/2}/\norm{U_1}_2$; note that using Lemma~\ref{lem lam large}
and by the assumption that $\lambda$ is large, $\gamma\le 1$ and so \eqref{eq cnw}
can be applied. We have
\[
\norm{U_1^\top S^\top S(Ax^*-b)}
	\le ( \sqrt{\eps/2}/\norm{U_1}_2) \norm{U_1}_2 \norm{Ax^*-b}
	\le \sqrt{\eps\Delta_*/2},
\]
and 
\[
\sr(U_1)\log(\sr(U_1))/\gamma^2
	\le \frac{\norm{U_1}_F^2}{\norm{U_1}_2^2} [ 2 \log(\norm{U_1}_F/\eps)] [ 2\norm{U_1}_2^2/\eps]
	= 4\norm{U_1}_F^2\log(\norm{U_1}_F/\eps)/\eps.
\]

Noting that since $Ax^*-b$ is a vector, its stable rank is one, we have that $m$ as claimed suffices for \eqref{eq prod}.
With \eqref{eq prod U1} and \eqref{eq prod}, the
claim for $\tilde{x}$ from an SRHT follows using Lemma~\ref{lem reg}.

The claim for (iii) follows as (ii), with a slightly simpler expression for $m$.
\end{proof}
} %\corsizeofSproof

\ifSUB
\begin{proof}
Please see \S\ref{subsec corsizeofSproof}.
\end{proof}
\else
\corsizeofSproof
\fi

%
%\begin{lemma}
%When $S$ is a sparse embedding, SRHT, or i.i.d. subgaussian as in Cor.~\ref{cor size of S},
%and $W$ is an $n\times d'$ matrix, then
%with constant probability $\norm{SW}_F\le K \norm{W}_F$ and $\norm{SW}_2 \le K \norm{W}$
%for a constant $K$.
%\end{lemma}
%
%\begin{proof} Please see Appendix A.3 of \cite{cnw}.
%\end{proof}
Here we mention the specific case of composing a sparse embedding matrix with
an SRHT.
 \begin{theorem}\label{thm reg stacked}
Given $A\in\R^{n\times d}$,
there are dimensions within constant factors of those given in Cor.~\ref{cor size of S}
such that for $S_1$ a sparse embedding and $S_2$ an SRHT with those dimensions,
\[
\tx\equiv \argmin_{x\in\R^d} \norm{S_2 S_1(Ax-b)}^2 + \lambda\norm{x}^2,
\]
satisfies
\[
\norm{A\tx-b}^2 + \lambda\norm{\tx}^2 \le (1+\eps) \min_{x\in\R^d} \norm{Ax-b}^2 + \lambda\norm{x}^2
\]
with constant probability.

Therefore in
\[
O(\nnz(A)) + \tO(d \sd_\lambda(A)/\eps + \sd_\lambda(A)^2)
\]
time, a ridge regression problem with $n$ rows can be reduced to one with
\[
O(\eps^{-1}(\sd_\lambda(A)+\log(1/\eps))\log(\sd_\lambda(A)/\eps))
\]
rows, whose solution is a $(1+\eps)$-approximate solution.
\end{theorem}

\def\thmregstackedproof{
\begin{proof}
This follows from Corollary~\ref{cor size of S} and the general comments of Appendix A.3 of \cite{CNW}; 
the results there imply that $\norm{S_i U_1}_F = \Theta(\norm{U_1}_F)$
and $\norm{S_i U_1}_2 = \Theta(\norm{U_1}_2)$ for $i\in[3]$ with constant probability,
which implies that $\sr(S_1U_1)$ and $\sr(S_2S_1U_1)$ are $O(\sr(U_1))$.
Moreover, the approximate multiplication bounds of \eqref{eq prod U1} and \eqref{eq prod}
have versions when using $S_2S_1U_1$ and ${S_2S_1(Ax^*-b)}$ to estimate products involving $S_1U_1$
and $S_1(Ax^*-b)$, so that for example, using the triangle inequality,
\begin{align*}
\norm{U_1^\top S_1^\top S_2^\top S_2 S_1 U_1 - U_1^\top U_1}_2
	  & \le \norm{U_1^\top S_1^\top S_2^\top S_2 S_1 U_1 - U_1^\top S_1^\top S_1 U_1}_2
		+ \norm{U_1^\top S_1^\top S_1 U_1 - U_1^\top U_1}_2
	\\ & \le 1/8 + 1/8 = 1/4.
\end{align*}
We have that $S = S_2 S_1$
satisfies \eqref{eq prod U1} and \eqref{eq prod}, as desired.
\end{proof}
}

\ifSUB
\begin{proof}
Please see \S\ref{subsec thmregstacked}.
\end{proof}
\else
\thmregstackedproof
\fi

Similar arguments imply that a reduction also using a sketching matrix $S_3$ with
subgaussian entries could be used, to reduce to a ridge regression problem with $O(\eps^{-1}\sd_\lambda(A))$
rows.

%\begin{theorem}
%There is an algorithm that in
%\[
%O(\nnz(A)\log\sd_\lambda(A)) + O(d(\sd_\lambda(A)/\eps + \sd_\lambda(A)^2))\poly(nd/\eps)
%\]
%expected time finds $\tilde{x}$ that is a $(1+\eps)$-approximate solution to \eqref{eq ridge1}.
%\end{theorem}
%
%\begin{proof}
%
%
%\end{proof}

\def\Kr{K}
\Ken{attempt at convergence test, "in progress":}
\Haim{If we can get it ``working'' we may want to demote the previous theorem from ``main result'' and have another
theorem here}
%DRAFT

%
%
%
%For $\tx$ as in Theorem~\ref{thm reg}, we have
%\begin{align*}
%\frac12 (\tx-x^*)^\top \nabla g(\tx)
%	   & = \norm{\tA(\tx-x^*)}^2
%	\\ & = \norm{\tA \tx - \hat{b}}^2 - \norm{\tA x^* - \hat{b}}^2
%	\\ & \le (1+\eps)\Delta_* - \Delta_*
%	\\ & = \eps\Delta_*.
%\end{align*}

%
%The quantity $\sd_\lambda(A)\le \sd_0(A) = \rank(A)$ is the \emph{statistical dimension} or \emph{effective dimension}
%of the regularized regression problem. Putting
%the corollary and the lemma together, sparsest embeddings ($\nnz(S)=n$) with
%$m=O(\sd_\lambda(A)^2 + \sd_\lambda(A)/\eps)$ are enough to obtain an $\eps$-approximate
%solution.

%@ARTICLE{2014arXiv1411.0306E,
%   author = {{El Alaoui}, A. and {Mahoney}, M.~W.},
%    title = "{Fast Randomized Kernel Methods With Statistical Guarantees}",
%  journal = {ArXiv e-prints},
%archivePrefix = "arXiv",
%   eprint = {1411.0306},
% primaryClass = "stat.ML",
% keywords = {Statistics - Machine Learning, Computer Science - Learning, Statistics - Computation},
%     year = 2014,
%    month = nov,
%   adsurl = {http://adsabs.harvard.edu/abs/2014arXiv1411.0306E},
%  adsnote = {Provided by the SAO/NASA Astrophysics Data System}
%}

\subsection{Large d}\label{subsec large d}

If the number of columns is larger than the number of rows, it is more attractive to 
sketch the rows, i.e., to use $A S^\top$. In general, we can express \eqref{eq ridge1} as
\[
\min_{x\in\R^d}  \norm{Ax}^2 - 2b^\top Ax + \norm{b}^2 + \lambda\norm{x}^2.
\]
We can assume $x$ has the form $x=A^\top y$, yielding the equivalent problem
\begin{equation}\label{eq no sketch}
\min_{y\in\R^n} \norm{AA^\top y}^2 - 2b^\top A A^\top y + \norm{b}^2 + \lambda\norm{A^\top y}^2.
\end{equation}
Sketching $A^\top$ with $S$ in the first two terms yields 
\begin{equation}\label{eq with  sketch}
\ty \equiv \argmin_{y\in\R^n} \lambda\norm{SA^\top y}^2 + \norm{AS^\top S A^\top y}^2 - 2b^\top A A^\top y + \norm{b}^2
\end{equation}
Now let $c^\top \equiv b^\top A A^\top$. Note that we can compute $c$ in $O(\nnz(A))$ time.
The solution to \eqref{eq with sketch} is, for $B\equiv SA^\top$ with $B^\top B$ invertible,
$\ty = (\lambda B^\top B + B^\top B B^\top B)^+ c/2$.

In the main result of this subsection, we show that provided $\lambda > 0$ then a sufficiently
tight subspace embedding to $\range(A^\top)$ suffices. 
\begin{theorem}\label{thm:twoProp}
Suppose $A$ has rank $k$,
 and its SVD is $A=U\Sigma V^\top$, with
$U\in\R^{n\times k}$, $\Sigma\in\R^{k\times k}$ and $V\in\R^{d\times k}$.
If $S\in\R^{m\times d}$ has
\begin{enumerate}
\item (Subspace Embedding) 
$E\equiv V^\top S^\top S V - \Iden_k$ with $\norm{E}_2\le \eps/2$
\item (Spectral Norm Approximate Matrix Product) 
for any fixed matrices $C, D$, each with $d$ rows, 
 $$\|C^T S^T S D - C^TD\|_2 \leq \eps' \|C\|_2 \|D\|_2,$$
where $\eps' \equiv (\eps/2)/(1+3\sigma_1^2/\lambda)$. 
%, and $\eps' \leq 1/3$ assuming
%%$\lambda$ is not large. 
\end{enumerate}
%
%\eps'\equiv $,
Then \eqref{eq with sketch}
has $\tx \equiv A^\top \ty$ approximately solving \eqref{eq ridge1}, that is,
\[
 \norm{A\tx - b}^2 +\lambda\norm{\tx}^2\le (1 + \eps) \Delta_*.
\]
%The same approximation result holds when $\norm{E}_2\le\eps''\equiv \eps/(1+2\sigma_1^2/\lambda + \eps\sum_i \sigma_i^2/\lambda)$.
\end{theorem}
%\Ken{Better to use $AS^\top S A^\top$ close to $AA^\top$, because that would allow $S$ to sample, but haven't gotten that to work}

\def\thmtwoPropproof{
\begin{proof}
To compare the sketched with the unsketched formulations, let $A$ have full SVD $A=U\Sigma V^\top$,
and let $w=\Sigma U^\top y$.
Using $\norm{Uz}=\norm{z}$ and $\norm{Vw}=\norm{w}$
yields the unsketched problem
\begin{equation}\label{eq no sk w}
\min_{w\in\R^k} \norm{\Sigma w}^2 -  2b^\top A V w + \norm{b}^2 + \lambda\norm{w}^2,
\end{equation}
equivalent to \eqref{eq no sketch}. The corresponding sketched version is
\[
\min_{w\in\R^k} \norm{\Sigma V^\top S^\top S V w}^2 -  2b^\top A V w + \norm{b}^2 +  \lambda\norm{S V w}^2.
\]

Now suppose $S$ has $E$ satisfying the first property in the theorem statement.
This implies $S$ is an $\eps/2$-embedding for $V$:
\[
| \norm{SVw}^2 - \norm{w}^2 |
	= | w^\top(V^\top S^\top S V - \Iden_k) w |
	\le (\eps/2)\norm{w}^2,
\]
and, using the second property in the theorem statement with $C^T = \Sigma V^T$ and $D = V$
(which do not depend on $w$),  
$$\norm{\Sigma V^\top S^\top S V - \Sigma}_2 = f,$$
where $f$ satisfies $|f| \leq \eps' \sigma_1$. It follows by the triangle inequality
for any $w$ that 
$$\norm{\Sigma V^\top S^\top SVw} \in [\norm{\Sigma w} - f\norm{w}, \norm{\Sigma w} + f\norm{w}].$$
Hence,
\begin{align*}
| \norm{\Sigma V^\top S^\top S V w}^2 - \norm{\Sigma w}^2 |
           & \in |(\norm{\Sigma w} \pm f \norm{w})^2 - \norm{\Sigma w}^2|
        \\ & \leq 2f \norm{\Sigma w} \norm{w} + f^2 \norm{w}^2\\
        \\ & \leq 3 \eps' \sigma_1^2 \norm{w}^2
\end{align*}
The value of \eqref{eq no sk w} is at least $\lambda\norm{w}^2$, so the relative error
of the sketch is at most
\[
\frac{\lambda (\eps/2)\norm{w}^2 + 3 \eps'\sigma_1^2 \norm{w}^2}{\lambda\norm{w}^2} \leq \eps.
\]
The statement of the theorem follows. 
%A similar argument shows that when $\norm{E}_2\le\eps''$,
%$S$ is still an embedding for $V$, and
%\[
%| \norm{\Sigma V^\top S^\top S V w}^2 - \norm{\Sigma w}^2 |
%	\le 2\eps''\sigma_1^2 \norm{w}^2 + (\eps'')^2\norm{w}^2\sum_i \sigma_i^2,
%\]
%and the second statement of the theorem follows.
\end{proof}
} % thmtwoPropproof

\ifSUB
\begin{proof}
Please see \S\ref{subsec thmtwoPropproof}.
\end{proof}
\else
\thmtwoPropproof
\fi

We now discuss which matrices $S$ can be used in Theorem \ref{thm:twoProp}. Note that
the first property is just the oblivious subspace embedding property, and we can 
use CountSketch, Subsampled Randomized Hadamard Transform, 
or Gaussian matrices to achieve this. One can also use OSNAP matrices \cite{nn13}; note
that here, unlike for Corollary \ref{cor size of S},
the running time will be $O(\nnz(A)/\epsilon)$ 
(see, e.g., \cite{Wbook} for a survey).  
For the second property, we use the
recent work of \cite{CNW}, where tight bounds for a number of oblivious
subspace embeddings $S$ were shown. 

In particular, applying the result in Appendix A.3 of \cite{CNW}, 
it is shown that the {\it composition} of matrices each satisfying the second property,
results in a matrix also satisfying the second property. It follows that we can let
$S$ be of the form $\Pi \cdot \Pi'$, where $\Pi'$ is an $r \times d$
CountSketch matrix, where $r = O(n^2/(\epsilon')^2)$, and $\Pi$ is an 
$\tilde{O}(n/(\epsilon')^2) \times r$ Subsampled Randomized Hadamard Transform. 
By standard results on oblivious subspace embeddings,
the first property of Theorem \ref{thm:twoProp} holds provided $r = \Theta(n^2/\epsilon^2)$
and $\Pi$ has $\tilde{O}(n/\epsilon^2)$ rows. Note that $\epsilon' \leq \epsilon$, so in
total we have $O(n/(\epsilon')^2)$ rows. 

Thus, we can compute $B = \Pi \cdot \Pi' A^T$ in 
$O(\nnz(A)) + \tilde{O}(n^3/(\epsilon')^2)$ time, 
and $B$ has $\tilde{O}(n/(\epsilon')^2)$ rows and $n$
columns. We can thus compute $\ty$ as above
in $\tilde{O}(n^3/(\epsilon')^2)$ additional time. 
Therefore in $O(\nnz(A)) + \tilde{O}(n^3/(\epsilon')^2)$
time, we can solve the problem of \eqref{eq ridge1}. 

We note that, using our results in Section \ref{sec large n}, 
in particular Theorem \ref{thm reg stacked}, 
we can first replace $n$ in the above time complexities 
with a function of 
$\sd_\lambda(A)$ and $\eps$, which can further reduce the overall
time complexity.

% !TEX root = p.tex

\subsection{Multiple-response Ridge Regression}

In multiple-response ridge regression one is interested in finding
\[
X^*\equiv \argmin_{X\in\R^{d\times {d'}}} \norm{AX-B}_F^2 + \lambda\norm{X}_F^2,
\]
where $B\in\R^{n\times d'}$.

It is straightforward to extend the results and algorithms for large $n$
to multiple regression. Since we use these results when we consider
regularized low-rank approximation, we state them next. The proofs
are omitted as they are entirely analogous to the proofs in subsection~\ref{sec large n}.
\begin{lemma}\label{thm reg mr}
Let $A$, $U_1$, $U_2$ as in Lemma~\ref{lem reg}, $B\in\R^{n\times {d'}}$,
\[
X^*\equiv \argmin_{X\in\R^{d\times {d'}}} \norm{AX-B}_F^2 + \lambda\norm{X}_F^2,
\]
and $\Delta_* \equiv \norm{AX^*-B}_F^2 + \lambda\norm{X^*}_F^2$.
Let sketching matrix $S\in\R^{m\times n}$ have a distribution such that with constant probability,
\begin{equation}\label{eq embed mr}
\norm{U_1^\top S^\top S U_1 - U_1^\top U_1}_2\le 1/4,
\end{equation}
and
\begin{equation}\label{eq prod mr}
\norm{U_1^\top S^\top S(B-AX^*) - U_1^\top (B-AX^*)}_F \le \sqrt{\eps \Delta_*}.
\end{equation}
Then
with constant probability,
\begin{equation}\label{eq mr sk}
\tX \equiv \argmin_{X\in\R^{d\times {d'}}} \norm{S(AX-B)}_F^2 + \lambda\norm{X}_F^2
\end{equation}
has $ \norm{A\tX-B}^2 + \lambda\norm{\tX}_F^2 \le (1+\eps)\Delta_*.$
\end{lemma}

% \begin{proof}
% Entirely analogous to that of Theorem~\ref{thm reg}.
% \end{proof}

%\begin{theorem}\label{thm reg mr stacked}
%There is $m=O(\sd_
%
%For $U_1$ as in the theorem, there is a constant $K>0$ so that there is
%\[
%m
%	\le K\max\{\norm{U_1}_F^4, \norm{U_1}_F^2/\eps\}
%	\le K(\sd_\lambda(A)^2 + \sd_\lambda(A)/\eps),
%\]
%so that if  $S\in\R^{m\times n}$ is a sparse embedding matrix, then
%the conditions \eqref{eq embed mr} and \eqref{eq prod mr} of the theorem apply, and $\tilde{X}$
%is an $\eps$-approximate solution to $\min_{x\in\R^d}  \norm{AX-B}_F^2 + \lambda\norm{X}^2$.
%\end{corollary}

\begin{theorem}\label{thm reg stacked mr}
There are dimensions  within a constant factor of those given in Thm.~\ref{thm reg stacked},
such that for $S_1$ a sparse embedding and $S_2$ SRHT with those dimensions,
$S=S_2S_1$ satisfies the conditions of Lemma~\ref{thm reg mr}, therefore the
corresponding $\tX$ does as well. That is, in time
\[
O(\nnz(A) + \nnz(B)) + \tO((d+{d'})(\sd_\lambda(A)/\eps + \sd_\lambda(A)^2)
\]
time, a multiple-response ridge regression problem with $n$ rows can be reduced to one with
%\[
$\tO(\eps^{-1}\sd_\lambda(A))$
%\]
rows, whose solution is a $(1+\eps)$-approximate solution.
\end{theorem}

% \begin{proof}
% Entirely analogous to that of Theorem~\ref{thm reg stacked}.
% \end{proof}

\begin{remark}\label{rem rowsp}
Note that the solution to \eqref{eq mr sk},
that is, the solution to $\min_X\norm{\hS(\hA X - \hB)}_F^2$,
where $\hS$ and $\hA$ are as defined in the proof of Lemma~\ref{lem reg}, and
$\hB\equiv \twomat{B}{0_{d\times {d'}}}$,
is $\tX = (\hS \hA)^+\hS\hB$; that is, the matrix $\hA \tX =  \hA(\hS \hA)^+\hS\hB$ whose
distance to $\hB$ is within $1+\eps$ of optimal
has rows in the rowspace of $\hB$, which is the rowspace of $B$. This property will
be helpful building low-rank approximations.
\end{remark}

%Due to space constraints, we refer the reader to the appendix for details and other results.

%\bibliographystyle{alpha}
%\bibliography{p}
%\appendix

\section{Ridge Low-Rank Approximation}\label{sec glra}

For an integer $k$ we consider the problem
\begin{equation}\label{eq lowr}
\min_{\substack{Y\in\R^{n\times k}\\ X\in\R^{k\times d}}}
		\norm{YX - A}_F^2 + \lambda\norm{Y}_F^2 + \lambda\norm{X}_F^2.
\end{equation}
From \cite{UHZB} (see also Corollary~\ref{cor simple sol} below),
this has the solution
\begin{equation}\begin{split}\label{eq shrink}
Y^*  &= U_k(\Sigma_k-\lambda\Iden_k)^{1/2}_+\\
X^*  &= (\Sigma_k-\lambda\Iden_k)^{1/2}_+ V_k^\top\\
\implies & \sd_\lambda(Y^*)  = \sd_\lambda(X^*)=\sum_{\substack{i\in [k]\\\sigma_i > \lambda}} (1-\lambda/\sigma_i)
\end{split}\end{equation}
where $U_k\Sigma_k V_k^\top$ is the best rank-$k$ approximation
to $A$, and for a matrix $W$, $W_+$ has entries that are equal to the
corresponding entries of $W$ that are nonnegative, and zero otherwise.

While \cite{UHZB} gives a general argument, it was also known (see for example \cite{srebro2005rank})
that when
the rank $k$ is large enough not to be an active constraint (say, $k=\rank(A)$),
then $Y^*X^*$ for $Y^*, X^*$ from \eqref{eq shrink} solves
\begin{equation*}
\min_{Z\in\R^{n\times d}} \norm{Z-A}_F^2 + 2\lambda\norm{Z}_*,
\end{equation*}
where $\norm{Z}_*$ is the nuclear norm of $X$ (also called the trace norm).

It is also well-known that
\[
\norm{Z}_* = \frac12(\min_{YX=Z} \norm{Y}_F^2 + \norm{X}_F^2),
\]
so that the optimality of \eqref{eq shrink} follows for large $k$.

%\Ken{some old proof for small $k$?}

\begin{lemma}\label{lem ZZ}
Given integer $k\ge 1$ and $\eps > 0$, $Y^*$ and $X^*$
as in \eqref{eq shrink},
there
are
\[ m = \tO(\eps^{-1}\sd_\lambda(Y^*)) = \tO(\eps^{-1}k)\text{ and } m'=\tO(\eps^{-1} \min\{k, \eps^{-1}\sd_\lambda(Y^*)\}),
\]
such that there is
a distribution on $S\in\R^{m\times n}$ and $R\in\R^{d\times m'}$
so that for
\[
Z^*_S, Z^*_R
	\equiv \argmin_{\substack{Z_S\in\R^{k\times m} \\Z_R\in\R^{m'\times k}}}
		\norm{ARZ_R Z_S SA - A}_F^2 + \lambda \norm{ARZ_R}_F^2 + \lambda\norm{Z_SSA}_F^2,
\]
with constant probability
$\tY \equiv ARZ^*_R$ and $\tX\equiv Z^*_S SA$ satisfy
\[
\norm{\tY\tX - A}_F^2 + \lambda\norm{\tY}_F^2 + \lambda\norm{\tX}_F^2
	\le (1+\eps)(\norm{Y^*X^* - A}_F^2 + \lambda\norm{Y^*}_F^2 + \lambda\norm{X^*}_F^2).
\]
The products $SA$ and $AR$ take altogether $O(\nnz(A)) + \tO((n+d)(\eps^{-2}\sd_\lambda(Y^*) + \eps^{-1}\sd_\lambda(Y^*)^2)$ to compute.
\end{lemma}

\begin{proof}
Let $Y^*$ and $X^*$ be an optimal solution pair for \eqref{eq lowr}.
Consider the problem
\begin{equation}\label{eq H}
\min_{H\in\R^{k\times d}} \norm{Y^*H - A}_F^2 + \lambda\norm{H}_F^2.
\end{equation}
Let $H^*$ be an optimal solution.
We can apply Lemma~\ref{thm reg mr}
mapping $A$ of the theorem to $Y^*$,
$B$ to $A$,
$Y^*$ to $H^*$,
and $\tY$ to $\tH\equiv \twomat{SY^*}{\sqrt{\lambda}\Iden_k}^+ \twomat{SA}{0_{k\times d}}$, so that for $S$ satisfying the condition
of Theorem~\ref{thm reg mr},
as noted in Remark~\ref{rem rowsp}, $\tH$ is within
$1+\eps$ of the cost of $H^*$, and in the rowspace of $SA$.
(That is, the rows of $\rowspan(\tH)\subset\rowspan(SA)$.)

Using Theorem~\ref{thm reg stacked mr}, we have $m =\tO(\eps^{-1}(\sd_\lambda(Y^*))=\tO(\eps^{-1}k)$.

Now consider the problem
\begin{equation}\label{eq toR}
\min_{W\in\R^{n\times k}} \norm{W\tH - A}_F^2 + \lambda\norm{W}_F^2.
\end{equation}
We again apply Lemma~\ref{thm reg mr},
mapping $A$ of the theorem to $\tH^\top$,
$B$ to $A^\top$,
$Y^*$ to the transpose of an optimal solution $W^*$ to \eqref{eq toR},
and $S^\top$ to a matrix $R$.
This results in $\tW \equiv \twomat{AR}{0_{k\times m'}}\twomat{\tH R}{\sqrt{\lambda}}^{{+}\top}$
whose cost is within $1+\eps$ of that of $W^*$.
(Here $Z^{{+}\top}$ denotes the transpose of the pseudo-inverse of $Z$.)
Moreover, the columns of $\tW$ are in the columspace of $AR$.

Since $\tH$ can be written in the form $Z_S SA$ for some $Z_S\in\R^{k\times m}$,
and $\tW$ in the form $ARZ_R$ for some $Z_R\in\R^{m'\times k}$,
the quality bound of the lemma follows, after adjusting $\eps$ by a constant factor.

Noting that $\rank(\tH)\le \min\{m,k\}$, there is big enough
\[
m'=\tO(\eps^{-1}\sd_\lambda(\tH)) = \tO(\eps^{-1}\min\{m,k\})
	=\tO(\min\{\eps^{-2}\sd_\lambda(Y^*), \eps^{-1}k\}).
\]
We apply Theorem~\ref{thm reg stacked mr} to obtain the time bounds for computing $SA$ and $AR$.
\end{proof}

We can reduce to an even yet smaller problem, using affine embeddings, which are
built using subspace embeddings. These are defined next.

\begin{definition}[subspace embedding]\label{def subs embed}
\mbox{}\\
Matrix $S\in\R^{m_S\times n}$ is a \emph{subspace $\eps$-embedding} for $A$ with respect to the Euclidean norm if
$\norm{SAx}_2 = (1\pm \eps)\norm{Ax}_2$ for all $x$.
\end{definition}

\begin{lemma}\label{lem subs embed}
There are \emph{sparse embedding} distributions on matrices $S\in\R^{m\times n}$
with $m=O(\eps^{-2}\rank(A)^2)$ so that $SA$ can
be computed in $\nnz(A)$ time, and with constant probability $S$ is a subspace $\eps$-embedding.
The SRHT (of Corollary~\ref{cor size of S}) is a distribution on $S\in\R^{m\times n}$ with
$m=\tO(\eps^{-2}\rank(A))$ such that $S$ is a subspace embedding with constant probability.
\end{lemma}

\begin{proof}
The sparse embedding claim is from \cite{cw13}, sharpened by \cite{nn13, mm13};
the SRHT claim is from for example \cite{bg12}.
\end{proof}

\begin{definition}[Affine Embedding] \label{def aff emb}
For $A$ as usual and $B\in\R^{n\times {d'}}$,
matrix $S$ is an \emph{affine $\eps$-embedding} for $A,B$ if
$\norm{S(AX-B)}_F^2 = (1\pm\eps)\norm{AX-B}_F^2$
for all $X\in\R^{d\times {d'}}$. A distribution over $\R^{m_S\times n}$ is a \emph{poly-sized affine embedding} distribution
if there is $m_S=\poly(d/\eps)$ such that constant probability, $S$ from the distribution is an affine $\eps$-embedding.
\end{definition}

\begin{lemma}\label{lem AE}
For $A$ as usual, $B\in\R^{n\times {d'}}$, suppose there is a distribution
over $S\in\R^{m\times n}$ so that with constant probability,
$S$ is a subspace embedding for $A$ with parameter $\eps$,
and for $X^*\equiv\argmin_{X\in\R^{d\times {d'}}} \norm{AX-B}_F^2$ and $B^*\equiv AX^* - B$,
$\norm{SB*}_F^2 = (1\pm\eps)\norm{B^*}_F^2$
and $\norm{U^\top S^\top SB^* - U^\top B^*} \le \eps\norm{B^*}_F^2$.
Then $S$ is an affine embedding for $A,B$.
A sparse embedding with $m=O(\rank(A)^2/\eps^2)$ has the needed properties. By first applying
a sparse embedding $\Pi$, and then a Subsampled Randomized Hadamard Transform (SHRT) $T$,
there is an affine $\eps$-embedding $S=T\Pi$ with $m=\tO(\rank(A)/\eps^2)$ taking time
$O(\nnz(A) + \nnz(B)) + \tO((d+{d'})\rank(A)^{1+\kappa}/\eps^2)$
time to apply to $A$ and $B$, that is, to compute $SA=T\Pi A$ and $SB$.
Here $\kappa > 0$ is any fixed value.
\end{lemma}

\begin{proof}
Shown in \cite{cw13}, sharpened with \cite{nn13, mm13}.
\end{proof}

\begin{theorem}\label{thm ZZ}
With notation as in Lemma~\ref{lem ZZ}, there
are
\[ p' = \tO(\eps^{-2}m) = \tO(\eps^{-3}\sd_\lambda(Y^*)) = \tO(\eps^{-3}k)\text{ and } p = \tO(\eps^{-2}m')=\tO(\eps^{-3} \min\{k, \eps^{-1}\sd_\lambda(Y^*)\}),
\]
such that there is
a distribution on $S_2\in\R^{p\times n}$, $R_2\in\R^{d\times p'}$
so that for
\[
\tZ_S,\tZ_R
	\equiv \argmin_{\substack{Z_S\in\R^{k\times m} \\Z_R\in\R^{m'\times k}}}
		\norm{S_2ARZ_R Z_S SAR_2 - S_2AR_2}_F^2
			+ \lambda \norm{S_2ARZ_R}_F^2 + \lambda \norm{Z_SSAR_2}_F^2,
\]
with constant probability
$\tY \equiv AR\tZ_R$ and $\tX\equiv \tZ_S SA$ satisfy
\[
\norm{\tY\tX - A}_F^2 + \lambda\norm{\tY}_F^2 + \lambda\norm{\tX}_F^2
	\le (1+\eps)(\norm{Y^*X^* - A}_F^2 + \lambda\norm{Y^*}_F^2 + \lambda\norm{X^*}_F^2).
\]
The matrices $S_2 AR$, $SAR$, and $SAR_2$ can be computed in $O(\nnz(A)) + \poly(\sd_\lambda(Y^*)/\eps)$
time.
\end{theorem}

\begin{proof}
Apply Lemma~\ref{lem AE}, with
$A$ of the lemma mapping to $AR$,
$B$ of the lemma mapping to $A$,
$U$ to the left singular matrix of $AR$,
$S$ to $S_2$, and
$d$ to $m'$.

Also apply Lemma~\ref{lem AE} in an analogous way, but in transpose, to $SA$.
For the last statement: to compute $SAR$, apply the sparse embedding of $S$ and
the sparse embedding of $R$ to $A$ on each side, and then the SRHT components
to the resulting small matrix; the claimed time bound follows. The other sketches
are computed similarly.
The theorem follows.
\end{proof}

\begin{lemma}\label{lem boyd}
For $C\in\R^{p\times m'}, D\in\R^{m\times p'}$, $G\in\R^{p\times p'}$, the problem of finding
\begin{equation}\label{eq lowr frob}
\min_{\substack{Z_S\in\R^{k\times m} \\Z_R\in\R^{m'\times k}}}
	\norm{CZ_R Z_S D - G}_F^2 +\lambda\norm{CZ_R}_F^2 + \lambda\norm{Z_SD}_F^2,
\end{equation}
and the minimizing $CZ_R$ and $Z_S D$, can be solved in
\[
O(pm' r_C + p'm r_D+ r_Dp(p' + r_C))
\]
time, where $r_C\equiv\rank(C)\le \min\{m',p\}$, and $r_D\equiv \rank(D) \le  \min\{m,p'\}$.
\end{lemma}

\begin{proof}
Let $U_C$ be an orthogonal basis for $\colspace(C)$,
so that every matrix of the form $CZ_R$ is equal to
$U_C Z'_R$ for some $Z'_R$. Similarly let $U_D^\top$ be
an orthogonal basis for $\rowspan(D)$, so that
every matrix of the form $Z_S D$ is equal to one of the
form $Z'_S U_D$. Let $P_C\equiv U_C U_C^\top$
and $P_D \equiv U_D U_D^\top$.
Then using $P_C(\Iden - P_C)=0$, $P_D(\Iden-P_D)=0$,
and matrix Pythagoras,
\begin{align*}
\norm{CZ_R Z_S D - G}_F^2 & +\lambda\norm{CZ_R}_F^2 + \lambda\norm{Z_SD}_F^2
	\\ & = \norm{P_C U_C Z'_R Z'_S U_D^\top P_D - G}_F^2 +\lambda\norm{U_C Z'_R}_F^2 + \lambda\norm{Z'_S U_D^\top}_F^2
	\\ & = \norm{P_C U_C Z'_R Z'_S U_D^\top P_D - P_C G P_D}_F^2
		+ \norm{(\Iden - P_C)G}_F^2
	\\ & \qquad + \norm{P_C G (\Iden - P_D)}_F^2
		+\lambda\norm{Z'_R}_F^2 + \lambda\norm{Z'_S}_F^2.
\end{align*}
So minimizing \eqref{eq lowr frob} is equivalent to minimizing
\begin{align*}
\norm{P_C U_C Z'_R Z'_S U_D^\top P_D  & - P_C G P_D}_F^2
		+\lambda\norm{Z'_R}_F^2 + \lambda\norm{Z'_S}_F^2
	\\ & = \norm{U_C Z'_R Z'_S U_D^\top - U_C U_C^\top G U_D U_D^\top}_F^2
		+\lambda\norm{Z'_R}_F^2 + \lambda\norm{Z'_S}_F^2
	\\ & =  \norm{Z'_R Z'_S - U_C^\top G U_D}_F^2
		+\lambda\norm{Z'_R}_F^2 + \lambda\norm{Z'_S}_F^2.
\end{align*}
This has the form of \eqref{eq lowr}, mapping $Y$ of \eqref{eq lowr}
to $Z'_R$, $X$ to $Z'_S$, and $A$ to $U_C^\top G U_D$,
from which a solution of the form \eqref{eq shrink} can be obtained.

To recover $Z_R$ from $Z'_R$: we have
$C=U_C \left[
	\begin{smallmatrix}
		T_C & T'_C
	\end{smallmatrix}
\right]$,
for matrices $T_C$ and $T'_C$,
where upper triangular $T_C\in\R^{r_C\times r_C}$.
We recover $Z_R$ as
$\left[\begin{smallmatrix} T_C^{-1} \hat Z'_R\\ 0_{m-r_C\times k}\end{smallmatrix}\right]$,
since then $U_C Z'_R = CZ_R$. A similar back-substitution allows recovery
of $Z_S$ from $Z'_S$.

Running times: to compute $U_C$ and $U_D$,
$O(pm' r_C + mp' r_D)$;
to compute $U_C^\top G U_D$, $O(r_D p(p' + r_C))$;
to compute and use the SVD of $U_C^\top G U_D$ to
to solve \eqref{eq lowr} via \eqref{eq shrink},
$O(r_C r_D \min\{r_C, r_D\})$;
to recover $Z_R$ and $Z_S$, $O(k(r_C^2 + r_D^2))$.
Thus, assuming $k\le \min\{p,p'\}$ and using $r_C\le\min\{p,m'\}$ and
$r_D\le\min\{m, p'\}$, the total running time is
$O(pm'r_C + p'mr_D+ pp'(r_C + r_D))$, as claimed.

%
%Running times: to compute the SVDs of $C$ and $D$,
%$O(pm'm_C + mp'm_D)$;
%to compute $\hA$, $O(r_Dp(p' + r_C))$;
%to compute and use the SVD of $\hA$ to use
%\eqref{eq shrink} to solve \eqref{eq reduced},
%$O(r_C r_D \min\{r_C, r_D\})$;
%to recover $Z_R$ and $Z_S$, $O(k(m'r_C + mr_D))$.
%Thus, assuming $k\le \min\{p,p'\}$ and using $r_C\le\min\{p,m'\}$ and
%$r_D\le\min\{m, p'\}$, the total running time is
%$O(pm'm_C + p'mm_D+ pp'(r_C + r_D))$, as claimed.
\end{proof}

\begin{theorem}\label{thm lowr ridge}
The matrices
$\tZ_S, \tZ_R$ of Theorem~\ref{thm ZZ} can be found in
$O(\nnz(A))+\poly(\sd_\lambda(Y^*)/\eps)$ time,
in particular $O(\nnz(A)) + \tO(\eps^{-7}\sd_\lambda(Y^*)^2\ \min\{k, \eps^{-1}\sd_\lambda(Y^*)\})$ time,
such that with constant probability,
$AR\tZ_R, \tZ_S SA$ is an $\eps$-approximate minimizer to \eqref{eq lowr}, that is,
\begin{align}
\norm{(AR\tZ_R) & ( \tZ_S SA) - A}_F^2 +  \lambda\norm{AR\tZ_R}_F^2 + \lambda\norm{ \tZ_S SA}_F^2
\\ & \le (1+\eps) \min_{\substack{Y\in\R^{n\times k}\\ X\in\R^{k\times d}}}
		\norm{YX - A}_F^2 + \lambda\norm{Y}_F^2 + \lambda\norm{X}_F^2.
\end{align}
With an additional $O(n+d)\poly(\sd_\lambda(Y^*)/\eps)$ time,
and in particular
\[
\tO(\eps^{-1}k\sd_\lambda(Y^*) (n+d + \min\{n,d\} \min\{k/\sd_\lambda(Y^*), \eps^{-1}\}))
\]
time, the solution matrices $\tY \equiv AR\tZ_R, \tX \equiv \tZ_S SA$ can be computed
and output.
\end{theorem}

\begin{proof}
Follows from Theorem~\ref{thm ZZ} and Lemma~\ref{lem boyd}, noting that for efficiency's sake we can use
the transpose of $A$ instead of $A$.
\end{proof}

%\section{Faster Unregularized Low-Rank Approximation}
%
%The approach of the previous section can be applied more efficiently when $\lambda=0$, since
%subspace embedding properties are no longer needed for the later sketching matrices $S_2$ and $R_2$
%of Theorem~\ref{thm ZZ}.
%
%\begin{theorem}
%For the problem
%\begin{equation}\label{eq lowr unreg}
%\min_{\substack{Y\in\R^{n\times k}\\ X\in\R^{k\times d}}}
%		\norm{YX - A}_F^2,
%\end{equation}
%matrices $\tY,\tX$ can be found in $O(\nnz(A) + (n+d)\eps^{-1}k + \poly(k/\eps)$ time,
%so that with constant probability, $\norm{\tY\tX - A}_F^2$ is within $1+\eps$ of optimal.
%\end{theorem}
%
%\begin{proof}
%As in the construction above,
%
%\end{proof}

% !TEX root = p.tex

\section{Regularized Canonical Correlation Analysis}
\label{sec:cca}

\cite{ABTZ14} showed how to use sketching to compute an approximate canonical correlation analysis (CCA).
In this section we consider a regularized version of CCA.

\begin{definition} Let $A\in\R^{n\times d}$ and $B\in\R^{n\times {d'}}$, and let
$q=\max(\rank(A^\top A + \lambda_1 \Iden_d), \rank(B^\top B + \lambda_2 \Iden_{d'}))$.
Let $\lambda_1 \geq 0$ and $\lambda_2 \geq 0$. The {\em $(\lambda_1, \lambda_2)$ canonical correlations }
$\sigma^{(\lambda_1, \lambda_2)}_1 \geq \dots \geq \sigma^{(\lambda_1, \lambda_2)}_q$ and
{\em $(\lambda_1, \lambda_2)$ canonical weights } $u_1,\dots,u_q\in \R^d$ and $v_1,\dots,v_q\in \R^{d'}$
are ones that maximize
$$
\tr(U^\top A^\top B V)
$$
subject to
\begin{eqnarray*}
  U^\top (A^\top A + \lambda_1 \Iden_d) U &=& I_q \\
  V^\top (B^\top B + \lambda_2 \Iden_{d'}) V &=& I_q \\
  U^\top A^\top B V &=& \diag(\sigma^{(\lambda_1, \lambda_2)}_1,\dots ,\sigma^{(\lambda_1, \lambda_2)}_q)
\end{eqnarray*}
where $U=\left[u_1,\dots,u_q \right]\in\R^{n\times q}$ and $V=\left[v_1,\dots,v_q \right]\in\R^{d'\times q}$.
\end{definition}

One classical way to solve non-regularized CCA ($\lambda_1 = \lambda_2 = 0$) is the Bj\"{o}rck-Golub algorithm~\cite{BG73}.
The regularized problem can be solved using a variant of that algorithm, as is shown in the following.

\begin{definition}
Let $A\in \R^{n \times d}$ with $n\geq d$ and let $\lambda \geq 0$. $A=QR$ is a {\em $\lambda$-QR factorization}
if $Q$ is full rank, $R$ is upper triangular and $R^\top R = A^\top A + \lambda \Iden_d$.
\end{definition}

\begin{remark}
A $\lambda$-QR factorization always exists, and $R$ will be invertible for $\lambda > 0$.
$Q$ has orthonormal columns for $\lambda = 0$.
\end{remark}

\begin{fact}
  For a $\lambda$-QR factorization $A=QR$ we have $Q^\top Q + \lambda R^{-\top} R^{-1} = \Iden_d$.
\end{fact}

\begin{proof}
A direct consequence of $R^\top R = A^\top A + \lambda \Iden_d$ (multiply from the right by $R^{-1}$ and the left by $R^{-\top}$).
\end{proof}

\begin{fact}
  For a $\lambda$-QR factorization $A=QR$ we have $\sd_\lambda(A) = \normF{Q}^2$.
\end{fact}

\begin{proof}
\begin{eqnarray*}
  \normF{Q}^2 = \tr(Q^\top Q) &=& \tr(\Iden_d - \lambda R^{-\top} R^{-1}) \\
   &=& d - \lambda \tr(R^{-\top} R^{-1}) \\
   &=& d - \lambda \tr((A^\top A + \lambda \Iden_d)^{-1}) \\
   &=& d - \sum_{i=1}^{d} \frac{\lambda}{\sigma^2_i + \lambda} \\
   &=& \sum_{i=1}^{d} \frac{\sigma^2_i}{\sigma^2_i + \lambda} \\
   &=& \sd_\lambda(A)\,.
\end{eqnarray*}
\end{proof}

\begin{theorem}[Regularized  Bj\"{o}rck-Golub]\label{thm reg bg}
Let $A=Q_A R_A$ be a $\lambda_1$-QR factorization of $A$, and $B=Q_B R_B$ be a $\lambda_2$-QR factorization of $B$.
Assume that $\lambda_1 > 0$ and $\lambda_2 > 0$.
The $(\lambda_1, \lambda_2)$ canonical correlations are exactly the singular values of $Q^\top_A Q_B$. Furthermore,
if $Q^\top_A Q_B = M \Sigma N^T$ is a thin SVD of $Q^\top_A Q_B$, then the columns of $R^{-1}_A M$ and $R^{-1}_B N$ are
canonical weights.
\end{theorem}

\begin{proof}
The constraints on $U$ and $V$ imply that $R_A U$ and $R_B V$ are orthonormal matrices, so the problem is equivalent to
maximizing $\tr(\tilde{U}^\top Q^\top_A Q_B \tilde{V})$ subject to $\tilde{U}$ and $\tilde{V}$ being orthonormal. A well-known result by Von Neumann (see~\cite{GZ95}) now implies that the maximum is bounded by the sum of the singular values of $Q^\top_A Q_B$ and that quantity is attained by setting $\tilde{U}=M$ and $M=\tilde{V}$. Simple algebra now establishes that
$U^\top A^\top B V = \Sigma$ and that the constraints hold.
\end{proof}

We now consider how to approximate the computation using sketching. The basic idea is similar to the one used in~\cite{ABTZ14} to accelerate the computation of non-regularized CCA:
compute the regularized canonical correlations and canonical weights of the pair $(SA,SB)$ for a sufficiently large
subspace embedding matrix $S$. Similarly to~\cite{ABTZ14}, we define the notion of approximate regularized CCA, and show
that for large enough $S$ we find an approximate CCA with high probability.

\begin{definition}[Approximate $(\lambda_1,\lambda_2)$ regularized CCA)]
  For $0\leq \eta \leq 1$, an $\eta$-approximate $(\lambda_1,\lambda_2)$ regularized CCA of $(A,B)$ is a set
  of positive numbers $\hat{\sigma}_1\geq \dots \geq \hat{\sigma}_q$, and vectors $\hat{u}_1,\dots,\hat{u}_q\in \R^d$ and $\hat{v}_1,\dots,\hat{v}_q\in \R^{d'}$ such that
  \begin{enumerate}[(a)]
    \item For every $i$,
    $$
    \left| \hat{\sigma}_i - \sigma_i^{(\lambda_1, \lambda_2)}\right| \leq \eta\,.
    $$
    \item Let  $\hat{U}=\left[\hat{u}_1,\dots,\hat{u}_q \right]\in\R^{n\times q}$ and
    	$\hat{V}=\left[\hat{v}_1,\dots,\hat{v}_q \right]\in\R^{d'\times q}$. We have,
        $$
        \left| \hat{U}^\top (A^\top A + \lambda_1 \Iden_d) \hat{U} - I_q\right| \leq \eta
        $$
        and
        $$
        \left| \hat{V}^\top (B^\top B + \lambda_2 \Iden_{d'}) \hat{V}s - I_q \right| \leq \eta\,.
        $$
        In the above, the notation $\left| X \right| \leq \alpha$ should be understood as
        entry-wise inequality.

    \item For every $i$,
    $$
    \left| \hat{u}^\top_i A^\top B \hat{v}_i - \sigma_i^{(\lambda_1, \lambda_2)}\right| \leq \eta\,.
    $$
  \end{enumerate}
\end{definition}

\begin{theorem}\label{thm approx reg cca}
If $S$ is a sparse embedding matrix with $m=\Omega(\max(\sd_{\lambda_1}(A), \sd_{\lambda_2}(B))^2 / \epsilon^2)$ rows, then with high probability the $(\lambda_1,\lambda_2)$ canonical correlations and canonical weights of $(SA,SB)$ form an  $\epsilon$-approximate $(\lambda_1,\lambda_2)$ regularized CCA for $(A,B)$.
\end{theorem}

\begin{proof}
We denote the approximate correlations and weights by $\hat{\sigma}_1\geq \dots \geq \hat{\sigma}_q$, $\hat{u}_1,\dots,\hat{u}_q\in \R^d$ and $\hat{v}_1,\dots,\hat{v}_q\in \R^{d'}$. Let  $\hat{U}=\left[\hat{u}_1,\dots,\hat{u}_q \right]\in\R^{n\times q}$ and $\hat{V}=\left[\hat{v}_1,\dots,\hat{v}_q \right]\in\R^{d'\times q}$.
Let $A=Q_A R_A$ be a $\lambda_1$-QR factorization of $A$, $B=Q_B R_B$ be a $\lambda_2$-QR factorization of $B$,
$SA=Q_{SA} R_{SA}$ be a $\lambda_1$-QR factorization of $SA$, and $B=Q_{SB} R_{SB}$ be a $\lambda_2$-QR factorization of $SB$. We use the notation $\sigma_i(\cdot)$ to denote the $i$th singular values of a matrix.

In the following we show that all three claims hold if the following three inequalities hold:
\begin{eqnarray*}
  \norm{Q^\top_{A} S^\top S Q_{B} - Q^\top_{A} Q_{B}}_F & \leq & \epsilon/2 \\
  \norm{Q^\top_{A} S^\top S Q_{A} - Q^\top_{A} Q_{A}}_F & \leq & \epsilon/4  \\
  \norm{Q^\top_{B} S^\top S Q_{B} - Q^\top_{B} Q_{B}}_F & \leq & \epsilon/4\,.
\end{eqnarray*}
Since for sparse embeddings it holds with high probability that
$$\norm{W^\top S^\top S H - W^\top H}_F \le C\norm{W}_F\norm{H}_F/\sqrt{m},$$
for a constant $C$, and since $\sd_{\lambda_1}(A)=\norm{Q_A}^2_F$ and  $\sd_{\lambda_2}(B)=\norm{Q_B}^2_F$,
all three will hold with high probability with $m$ that is large enough as in the theorem statement.

\paragraph{Proof of (a).} As a consequence of Theorem~\ref{thm reg bg}, we have
\begin{eqnarray*}
  \left| \hat{\sigma}_i - \sigma_i^{(\lambda_1, \lambda_2)}\right| &=& \left| \sigma_i(Q^\top_{SA} Q_{SB}) - \sigma_i(Q^\top_{A} Q_{B}) \right| \\
  &\leq& \left| \sigma_i(Q^\top_{SA} Q_{SB}) - \sigma_i(Q^\top_{A} S^\top S Q_{B}) \right| + \left| \sigma_i(Q^\top_{A} S^\top S Q_{B}) - \sigma_i(Q^\top_{A} Q_{B}) \right|
\end{eqnarray*}

It is always the case that $\left|\sigma_i(\Psi) - \sigma_i(\Phi) \right| \leq \norm{\Psi - \Phi}_2$~\cite[Corollary 7.3.5]{HJ13} so with high probability
\begin{eqnarray*}
  \left| \sigma_i(Q^\top_{A} S^\top S Q_{B}) - \sigma_i(Q^\top_{A} Q_{B}) \right| &\leq& \norm{Q^\top_{A} S^\top S Q_{B} - Q^\top_{A} Q_{B}}_2 \\
   &\leq& \norm{Q^\top_{A} S^\top S Q_{B} - Q^\top_{A} Q_{B}}_F \\
   &\leq& \epsilon/2\,.
\end{eqnarray*}

To bound $ \left| \sigma_i(Q^\top_{SA} Q_{SB}) - \sigma_i(Q^\top_{A} S^\top S Q_{B}) \right|$ we use the
fact~\cite[Theorem 3.3]{EI95} that for nonsingular $D_L$ and $D_R$ we have $\left|\sigma_i(\Psi) - \sigma_i(\Phi) \right| \leq \gamma \cdot \sigma_i(\Psi)$ for
$$
\gamma=\max(\norm{D_L D^\top_L - \Iden}_2, \norm{D^\top_R D_R - \Iden}_2)\,.
$$
Let $D_L= R^{-\top}_A R^\top_{SA}$ and $D_R=R_{SB}R^{-1}_B$. Both are nonsingular because $\lambda_1 > 0$ and $\lambda_2 > 0$. We now have
\begin{eqnarray*}
   \norm{D_L D^\top_L - \Iden}_2 &=& \norm{R^{-\top}_A R^\top_{SA} R_{SA} R^{-1}_A - \Iden}_2\\
   &=& \norm{R^{-\top}_A (A^\top S^\top S A + \lambda_1 I) R^{-1}_A  - \Iden}_2\\
   &=& \norm{Q_A^\top S^\top S Q_A + \lambda_1 R^{-\top}_A R^{-1}_A - \Iden}_2 \\
   &=& \norm{Q^\top_{A} S^\top S Q_{A} - Q^\top_{A} Q_A}_2 \\
   & \leq & \epsilon/4\,.
\end{eqnarray*}
Similarly, we bound $\norm{D^\top_R D_R - \Iden}_2\leq \epsilon/4$. We now have
\begin{eqnarray*}
\left| \sigma_i(Q^\top_{SA} Q_{SB}) - \sigma_i(Q^\top_{A} S^\top S Q_{B}) \right| & = &
\left| \sigma_i(R^{-1}_A A^\top S^\top S B R^{-1}_B) - \sigma_i(R^{-1}_{SA} A^\top S^\top S B R^{-1}_{SB}) \right| \\
& \leq & \epsilon/4 \cdot \sigma_i(R^{-1}_A A^\top S^\top S B R^{-1}_B) \\
& = & \epsilon/4 \cdot \sigma_i(Q^\top_A S^\top S Q^\top_B) \\
& \leq & \epsilon/4 \cdot \left(\sigma_i(Q^\top_A Q_B) + \left|\sigma_i(Q^\top_A Q_B) - \sigma_i(Q^\top_A S^\top S Q_B) \right|\right) \\
& \leq & \epsilon / 4 \cdot (1 + \epsilon/2) \\
& \leq & \epsilon / 2\,.
\end{eqnarray*}

\paragraph{Proof of (b).} We prove the claim for $\hat{U}$. The proof for $\hat{V}$ is analogous. We need to show that
with high probability $ \left| \hat{U}^\top (A^\top A + \lambda_1 \Iden_d) \hat{U} - I_q\right| \leq \epsilon$.
Note, that since $\hat{u}_1,\dots,\hat{u}_q$ are canonical weights of $(SA,SB)$, then we know that
$\hat{U}^\top (A^\top S^\top S A + \lambda_1 \Iden_d) \hat{U} = I_q$. So, the claim is equivalent to the claim
that for all $i,j$ we have
$$
\left| \hat{u}_i^\top (A^\top A + \lambda_1 \Iden_d) \hat{u}_j - \hat{u}_i^\top (A^\top S^\top S A + \lambda_1 \Iden_d) \hat{u}_j \right| \leq \epsilon\,.
$$
For all $i$, $j$, we have
\begin{eqnarray*}
  \left| \hat{u}_i^\top (A^\top A + \lambda_1 \Iden_d) \hat{u}_j - \hat{u}_i^\top (A^\top S^\top S A + \lambda_1 \Iden_d) \hat{u}_j \right| &=& \left| \hat{u}_i^\top R^\top_A R_A \hat{u}_j - \hat{u}_i^\top R^\top_{SA} R_{SA} \hat{u}_j \right| \\
   &=& \left| \hat{u}_i^\top \left(R^\top_A R_A -  R^\top_{SA} R_{SA} \right) \hat{u}_j \right| \\
  &=& \left| \hat{u}_i^\top R^\top_A \left( \Iden - R^{-\top}_A R^\top_{SA} R_{SA} R^{-1}_A \right) R_A \hat{u}_j \right|
\end{eqnarray*}

If $i=j$, the Courant-Fischer theorem now implies that
\begin{eqnarray*}
  \left| \hat{u}_i^\top (A^\top A + \lambda_1 \Iden_d) \hat{u}_i - 1 \right| &=&  \left| \hat{u}_i^\top R^\top_A \left( \Iden - R^{-\top}_A R^\top_{SA} R_{SA} R^{-1}_A \right) R_A \hat{u}_i \right|   \\
  &\leq & \norm{\Iden - R^{-\top}_A R^\top_{SA} R_{SA} R^{-1}_A}_2 \cdot \hat{u}_i^\top R^\top_A R_A \hat{u}_i \\
  &=& \norm{\Iden - R^{-\top}_A R^\top_{SA} R_{SA} R^{-1}_A}_2 \cdot \hat{u}_i^\top (A^\top A + \lambda_1 \Iden) \hat{u}_i \\
  &\leq& (\epsilon/4) \cdot \hat{u}_i^\top (A^\top A + \lambda_1 \Iden) \hat{u}_i\,.
\end{eqnarray*}
The last inequality is due to the fact that we already shown in the proof of (a) that $\norm{\Iden - R^{-\top}_A R^\top_{SA} R_{SA} R^{-1}_A}_2 \leq \epsilon /4 $. Therefore,
$$
\hat{u}_i^\top (A^\top A + \lambda_1 \Iden_d) \hat{u}_i \leq 1 + (\epsilon/4) \cdot \hat{u}_i^\top (A^\top A + \lambda_1 \Iden_d) \hat{u}_i
$$ so
$$
\hat{u}_i^\top (A^\top A + \lambda_1 \Iden_d) \hat{u}_i \leq \frac{1}{1-\epsilon/4}\leq 2\,.
$$ which now implies that $ \left| \hat{u}_i^\top (A^\top A + \lambda_1 \Iden_d) \hat{u}_i - 1 \right| \leq \epsilon/2$.

For $i\neq j$, the submultiplicativity property of matrix norms implies that
\begin{eqnarray*}
\left| \hat{u}_i^\top (A^\top A + \lambda_1 \Iden_d) \hat{u}_j \right| & \leq &
\norm{\Iden - R^{-\top}_A R^\top_{SA} R_{SA} R^{-1}_A}_2 \cdot \norm{R_A \hat{u}_i}_2 \cdot \norm{R_A \hat{u}_i}_2 \\
& \leq & (\epsilon/4) \cdot \sqrt{\hat{u}_i^\top (A^\top A + \lambda_1 \Iden) \hat{u}_i} \cdot \sqrt{\hat{u}_j^\top (A^\top A + \lambda_1 \Iden) \hat{u}_j} \\
& \leq & (\epsilon/4) \cdot \max(\hat{u}_i^\top (A^\top A + \lambda_1 \Iden) \hat{u}_i, \hat{u}_j^\top (A^\top A + \lambda_1 \Iden) \hat{u}_j) \\
& \leq & \epsilon / 2
\end{eqnarray*}

\paragraph{Proof of (c).} It is enough to show (after adjusting constants) that
$$
\left| \hat{u}^\top_i A^\top B \hat{v}_i - \hat{\sigma}_i \right| \leq \epsilon
$$
since we already shown that $ \left| \hat{\sigma}_i - \sigma_i^{(\lambda_1, \lambda_2)}\right| \leq \epsilon$.
We have,
\begin{eqnarray*}
  \left| \hat{u}^\top_i A^\top B \hat{v}_i - \hat{\sigma}_i \right| &=& \left| \hat{u}^\top_i A^\top B \hat{v}_i - \hat{u}^\top_i A^\top S^\top S B \hat{v}_i \right| \\
  &=& \left| \hat{u}^\top_i R^\top_A (Q^\top_A Q_B - Q^\top_A S^\top S Q_B) R_B \hat{v}_i\right| \\
  &\leq& \norm{Q^\top_A Q_B - Q^\top_A S^\top S Q_B}_2 \cdot \norm{R_A \hat{u}_i}_2 \cdot \norm{R_B \hat{v}_i}_2 \\
  &\leq& \norm{Q^\top_A Q_B - Q^\top_A S^\top S Q_B}_2 \cdot \max(\hat{u}_i^\top (A^\top A + \lambda_1 \Iden) \hat{u}_i, \hat{v}_i^\top (B^\top B + \lambda_2 \Iden) \hat{v}_i) \\
  & \leq & \epsilon/2
\end{eqnarray*}
\end{proof}

Taking an optimization point of view, the following Corollary shows that the suboptimality in the objective is not too big (the fact that the constraints are approximately held is established in the previous theorem).
\begin{corollary}
  Let $U_L$ and $V_L$ (respectively, $\hat{U}_L$ and $\hat{V}_L$) denote the first $L$ columns of $U$ and $V$ (respectively, $\hat{U}$ and $\hat{V}$. Then,
  $$
  \tr(\hat{U}_L^\top A^\top B \hat{V}_L) \leq \tr(U_L^\top A^\top B V_L) + \epsilon L\,.
  $$
\end{corollary}

% !TEX root = p.tex

\section{General Regularization: Multiple-response Regression}\label{sec reg mr gen}

In this section we consider the problem
$$
X^*\equiv \argmin_{X\in\R^{d\times {d'}}} \norm{AX-B}_F^2 + f(X)
$$
for a real-valued function $f$ on matrices. We show that under certain
assumptions on $f$ (generalizing from $f(X)=\norm{X}_h$ for some orthogonally invariant norm $\nm{h}$),
if we have an approximation algorithm for the problem,
then via sketching the running time dependence of the algorithm on $n$ can be improved.

\begin{definition}[contractions, reduction by contractions]
A square matrix $P$ is a \emph{contraction} if its spectral norm
$\norm{P}_2\le 1$.
Say that $f()$ is \emph{left reduced by contractions}
if $f(PA) \le f(A)$ for all $A$ and contractions $P$.
Similarly define \emph{right reduced by contractions}. Say that
$f()$ is \emph{reduced by contractions} if it is both left and right reduced by contractions.
\end{definition}

\begin{definition}[(left/right) orthogonal invariance(\loi/\roi)]
A matrix measure $f()$ is
\emph{left orthogonally invariant} (or \loi\ for short)
if $f(UA)= f(A)$ for all $A$ and orthogonal $U$.
Similarly define \emph{right orthogonal invariance} (\roi).
Note that $f()$ is orthogonally invariant if it is both left and right orthogonally invariant. 
\end{definition}

When a norm $\nm{g}$ is orthogonally invariant, it can be expressed as $\norm{A}_g = g(\sigma_1, \sigma_2,\ldots,\sigma_r)$,
where the $\sigma_i$ are the singular values of $A$, and
$g()$ is a \emph{symmetric gauge function}: a function that is even in each argument, and symmetric, meaning that its value depends
only on the set of input values and not their order.

\begin{lemma}
If $P$ is a contraction, then $P$ is a convex combination of orthogonal matrices:
$P=\sum_j \alpha_j U_j$, where each $U_j$ is orthogonal,
$\sum_j \alpha_j = 1$,  and $\alpha_j\ge 0$.
\end{lemma}

\begin{proof} Please see \cite{horn1994topics}, exercise 3.1.5(h).
Briefly: the vector of singular values
is contained in the hypercube $[-1,1]^n$, and so is a convex combination of $n+1$ hypercube vertices; as diagonal matrices, these
are orthogonal matrices, so if $P$ has SVD $P=U\Sigma V^\top$, then $\Sigma=\sum_j \alpha_j D_j$, where each
$D_j$ is an orthogonal diagonal matrix, and so $P=U (\sum_j \alpha_j D_j) V^\top = \sum_j \alpha_j U D_j V^\top$; each summand
is an orthogonal matrix.
\end{proof}

\begin{lemma}\label{lem norm proj}\cite{domon2006}
If matrix measure $f()$ is left orthogonally invariant and subadditive, then it is left reduced by contractions,
and similarly on the right.
\end{lemma}

\begin{proof}
(Given here for convenience.)
Using the representation of $P$ as a convex combination from the lemma just above,
$$f(PA) = f( \sum_j \alpha_j U_j A)
	\le \sum_j f(\alpha_j U_j A)
	= \sum_j \alpha_j f(U_j A)
	= f(A),
$$
and $f()$ is left reduced by contractions, as claimed.
\end{proof}

%
%
%\begin{definition}[Reduction by projection]
%Say that a matrix measure $f()$ is \emph{reduced by projection on the right} if
%for any projection matrix $P$ and matrix $A$, $\norm{AP}_h\le \norm{A}_h$.
%Say that $f()$ is \emph{reduced by projection on the left}
%if $\norm{PA}_h\le\norm{A}_h$ for all $A$ and projections $P$, and
%say that $f()$ is \emph{reduced by projections} if it is reduced by
%projections on both the left and right.
%\end{definition}

%\begin{lemma}\label{lem norm proj}
%If matrix matrix measure $f()$ is orthogonally invariant, then $f()$ is reduced by projections.
%\end{lemma}
%
%\begin{proof}
%Let $P$ be a projection matrix applicable on the left.
%Since $P$ is a projection, $P^2=P$, so that $\norm{Px}_2^2 = x^\top P^2 x = x^\top P x\le \norm{x}\norm{Px}$,
%so $\norm{Px}\le \norm{x}$. From the Courant-Fischer-Weyl min-max principle, the singular values
%\[
%\sigma_i(PA)
%	= \min_{\dim(L)=i}\max_{\substack{x\in L\\ \norm{x}=1}} \norm{PAx}
%	\le  \min_{\dim(L)=i}\max_{\substack{x\in L\\ \norm{x}=1}} \norm{Ax}
%	=\sigma_i(A),
%\]
%where the minimization is over subspaces $L$.
%This implies that $f()$ is reduced by projections on the
%left, using the orthogonal invariance of $f()$
%and the Ky Fan dominance theorem \cite{fan1951maximum}.  Reduction by projections on the right follows similarly,
%or via transpose, and the claim follows by definition.
%\end{proof}

\begin{definition}\label{def v-norm}
Fix $p\ge 1$. The \emph{$v$-norm} of matrix $A$ is $\norm{A}_v\equiv \left[\sum_{i\in [d]} \norm{A_{i:}}_2^p\right]^{1/p}$.
\end{definition}

This is also called the $(2,p)$-norm \cite{UHZB} or $R^1$ norm when $p=1 $\cite{dzhz06}. 

\begin{remark}\label{rem reduced norms}
Since $\nm{F}$, the spectral norm $\nm{2}$, and the trace norm $\nm{*}$ are orthogonally invariant, they are 
reduced by contractions.
Some $f()$ are reduced by contractions on one side,
without being orthogonally invariant: for example, the $v$-norm $\nm{v}$ is right orthogonally invariant,
and therefore by Lemma~\ref{lem norm proj}, right reduced by contractions, but not on the left.

The $v$-norm can also be considered for $p<1$; this is not subadditive, and so
Lemma~\ref{lem norm proj} does not apply, but even so, it is right orthogonally invariant and
right reduced by contractions, just considering the invariance or contractions row-wise.
\end{remark}

\begin{definition}[subspace embedding w.r.t. a matrix norm, poly-sized distributions]\label{def subs embed gen}
\mbox{}\\
From Definition~\ref{def subs embed},
a matrix $S\in\R^{m_S\times n}$ is a \emph{subspace $\eps$-embedding} for $A$ with respect to the Euclidean norm if
$\norm{SAx}_2 = (1\pm \eps)\norm{Ax}_2$ for all $x$. More generally, $S$ is a (left) subspace $\eps$-embedding for $A$ with
respect to a matrix measure $f()$ if $f(SAX)  = (1\pm\eps) f(AX)$ for all matrices $X$. Say that $R\in\R^{d\times m_R}$ is a right  subspace $\eps$-embedding for $A$ with respect to $f()$
if $f(YAR) = (1\pm\eps)f(YA)$ for all matrices $Y$. Say that a probability distribution
over matrices $S$ is a \emph{poly-sized sketching distribution} if there is $m_S=\poly(d/\eps)$ so that
with constant probability, $S$ is a subspace $\eps$-embedding.
Similarly define for sketching on the right, where the size condition on $m_R$ is $m_R=\poly(n/\eps)$.
\end{definition}

\begin{definition}[padding invariance]\label{def pad invar}
Say that a matrix measure $f()$ is \emph{padding invariant} if it is preserved by
padding $A$ with rows or columns of zeroes:
$f(\twomat{A}{0_{z\times d}})= f\left(\begin{smallmatrix}A & 0_{n\times z'}\end{smallmatrix}\right) = f(A)$.
\end{definition}

\begin{lemma}\label{lem padding}
Unitarily invariant norms and $v$-norms are padding invariant.
\end{lemma}

\begin{proof}
For $v$-norms, this is direct from the definition. For unitarily invariant norms, this follows from their dependence on the singular values only,
and that the singular values of a matrix don't change with padding: if $A=U\Sigma V^\top$,
then the SVD of $\twomat{A}{0}$ is $\twomat{U}{0}\Sigma V^\top$, and correspondingly for column padding.
\end{proof}

\begin{definition}[\piloi, \piroi]
Say that a matrix measure is \piloi\ if it is padding invariant and left orthogonally invariant, and \piroi\ if it is
padding invariant and right orthogonally invariant.
\end{definition}

\begin{definition}[embedding inheritance]
Say that a matrix measure $f()$ \emph{inherits} a subspace $\eps$-embedding from the Euclidean norm (on the left)
if the condition that $S\in\R^{m\times n}$ is a subspace $\eps$-embedding for $A$
with respect to the Euclidean norm implies that $S$ is a subspace $O(\eps)$-embedding for $f()$.
Define inheritance on the right similarly.
\end{definition}

\begin{lemma}\label{lem UI->embed}
If matrix measure $f()$ is \piloi\ then it inherits a left subspace $\eps$-embedding
from the Euclidean norm, and similarly on the right.
\end{lemma}

\begin{proof}
Since the columns of $AY$ are members of the columnspace of $A$, they can be expressed in terms
of a basis for that columnspace; that is there is $U$ with orthonormal columns so that for any $AY$
there is some $Z$ so that $AY=UZ$.
\Ken{using obliviousness here...} So we will assume that $A$ has orthonormal columns.
Note that from padding invariance, if $n>d$ we can expand $A$ with orthonormal columns $\bar{A}$
so that $\twomatr{A}{\bar{A}}$ is an orthogonal matrix,
and pad $Y$ with zero rows, so that
\[
f(AY) = f(\twomatr{A}{\bar{A}}\twomat{Y}{0})= f(\twomat{Y}{0}) = f(Y)
\]

We need to show that for $S$ a subspace $\eps$-embedding for $A$,
it holds that $(1-O(\eps)) f(AY)\le f(SAY)\le (1+O(\eps)) f(AY)$ for all $Y$.
For the upper bound on $f(SAY)$,
since $\norm{SAx}_2\le (1+\eps)\norm{Ax}_2 = (1+\eps)\norm{x}_2$,
we know that $\norm{SA}_2\le 1+\eps$, so that $\frac{1}{1+\eps}SA$ is a (nonsquare) contraction.
Moreover, if we pad with zeros to make a square matrix, we do not not increase the spectral norm.
Since $f()$ is assumed padding invariant, if $SA$ is padded with zero columns, we can
also pad $Y$ with rows of zeros. Suppose $m>d$, so we pad $SA$ with $m-d$ zero columns,
and $Y$ with $m-d$ zero rows. 
So $\frac{1}{1+\eps}\twomatr{SA}{0}$ is a square contraction,
and from the left orthogonal invariance of $f()$ and Lemma~\ref{lem norm proj},
we have
\[
\frac{1}{1+\eps} f(SAY) = f(\frac{1}{1+\eps} \twomatr{SA}{0}\twomat{Y}{0}) \le f(\twomat{Y}{0}) = f(Y),
\]
and as noted above, $f(Y) = f(AY)$, so
that $f(SAY) \le (1+\eps)f(AY)$, as desired.

For the lower bound $f(SAY) \ge (1-\eps) f(AY)$:
since $\norm{SAx}_2 \ge (1-\eps)\norm{Ax}_2 = (1-\eps) \norm{x}_2$ for all $x$,
$\inf_{x} \norm{A^\top S^\top SAx}_2
	= \inf_{x} \norm{SAx}_2^2/\norm{x}^2
	\ge (1-\eps)^2$,
so that $\norm{(A^\top S^\top SA)^{-1}}_2\le 1/(1-\eps)^2$,
and $\norm{(A^\top S^\top SA)^{-1}A^\top S^\top}_2 \le (1+\eps)/(1-\eps)^2 \le 1+O(\eps)$.
Thus
\[
f(AY) = f(Y)
	= f((A^\top S^\top SA)^{-1} A^\top S^\top SA Y)
	\le (1+O(\eps)) f(SAY),
\]
and $f(SAY) \ge (1-O(\eps)) f(AY)$,
as claimed.
\end{proof}

\begin{remark}\label{rem v-norm embed}
Note that a sketching matrix that is a subspace $\eps$-embedding on the right
for the Euclidean norm is also a subspace embedding on the right for $\nm{v}$, even when $p<1$,
just applying the Euclidean embedding row-wise.
\end{remark}

%\begin{proof}
%For given matrix $Y$, we have $\norm{SAYx}_2 \le (1+\eps)\norm{AYx}_2$ for all $x$. Similarly to Lemma~\ref{lem norm proj},
%from the Courant-Fischer-Weyl min-max principle, this inequality implies that the singular values
%\[
%\sigma_i(SAY)
%	= \min_{\dim(L)=i}\max_{\substack{x\in L\\ \norm{x}=1}} \norm{SAYx}
%	\le  (1+\eps) \min_{\dim(L)=i}\max_{\substack{x\in L\\ \norm{x}=1}} \norm{AYx}
%	=(1+\eps)\sigma_i(AY) = \sigma_i((1+\eps)AY),
%\]
%which implies $\norm{SAY}_h\le \norm{(1+\eps)AY}_h = (1+\eps)\norm{AY}_h$
%using the Ky Fan dominance theorem.
%Similarly, $\sigma_i(SAY)\ge (1-\eps)\sigma_i(AY)$, which implies $\norm{SAY}_h \ge (1-\eps)\norm{AY}_h$,
%and the lemma follows.
%\end{proof}

%\Ken{Does the following define more than polynomials?}
%
%\begin{definition}[approximation preserving]\label{def pres approx}
%Say that a function $z : \R \rightarrow \R$ is \emph{approximation-preserving} if
%if there is a constant $C_z$ so that so that $z((1\pm\eps)a) = (1\pm C_z\eps) z(a)$
%for all $a\in\R$.
%\end{definition}
%
%\Ken{probably should re-do the below with $f()$ the matrix measure instead of $f(\norm{}_h)$}

\begin{lemma}\label{lem rowspace SB}
Let $f()$ be a real-valued function on matrices that is right orthogonally invariant, right reduced by contractions,
and inherits a sketching distribution from the Euclidean norm.
(If $f()$ is \piroi\ and subadditive, these conditions hold by Lemmas~\ref{lem norm proj} and \ref{lem UI->embed}.)
Let $B\in\R^{n\times {d'}}$.
Let
\begin{equation}\label{eq ABgf exact}
X^*\equiv \argmin_{X\in\R^{d\times {d'}}} \norm{AX-B}_F^2 + f(X),
\end{equation}
and $\Delta_* \equiv \norm{AX^*-B}_F^2 + f(X^*)$.
Let $S\in\R^{m_S\times n}$ for parameter $m_S$ be an affine $\eps$-embedding for $A,B$
with respect to $\nm{F}$. Then
\[
Z^*\equiv\argmin_Z \norm{AZSB - B}_F^2 + f(ZSB)
\]
has
\[
\norm{(AZ^*SB -B)}_F^2 + f(Z^*SB)
	\le (1+\eps)\Delta_*,
\]
\end{lemma}

\begin{proof}
Let $X^*_S\equiv \argmin_{X\in\R^{d\times {d'}}} \norm{S(AX-B)}_F^2 + f(X)$.
If $S$ is an affine embedding for $A,B$, then $X^*_S$ is a good approximate solution
to \eqref{eq ABgf exact}, that is,
$\norm{AX^*_S - B}_F^2 + f(X^*_S) \le (1+\eps)\Delta_*$.
Let $P_{SB}$ be the orthogonal projection onto $\rowspan(SB)$; note that $P_{SB}$ is a contraction.
Then by hypothesis,
\[
\norm{(SAX^*_S -SB)P_{SB}}_F^2 + f(X^*_S P_{SB})
	\le \norm{SAX^*_S -SB}_F^2 + f(X^*_S),
\]
using also that the Frobenius norm is reduced by contraction, as noted in Remark~\ref{rem reduced norms}.
That is, $X^*_S P_{SB}$ has
cost no higher than that of $X^*_S$, or put another way,
without loss of generality, $X^*_S$ has rows in $\rowspan(SB)$.
Since $X^*P_{SB}$ can be expressed as $ZSB$ for some
$Z$, the lemma follows.
\end{proof}

The following is the main theorem of this section.
\begin{theorem}\label{thm gen regul regr}
Let $f()$ be a real-valued function on matrices that is right orthogonally invariant, right reduced by contractions,
and inherits a sketching distribution from the Euclidean norm on the right.
(If $f()$ is \piroi\ and subadditive, these conditions hold by Lemmas~\ref{lem norm proj} and \ref{lem UI->embed}.)
Let $B\in\R^{n\times {d'}}$.
Let $X^*$ and $\Delta_*$ as in Lemma~\ref{lem rowspace SB}.
Suppose that for $r\equiv\rank A$,
there is an algorithm that for general $n,d,{d'},r$ and $\eps > 0$,
finds $\tX$ with $\norm{A\tX-B}_F^2 + f(\tX) \le (1+\eps)\Delta_*$
in time $\tau(d,n,{d'},r,\eps)$.  Then there is an algorithm that with constant
probability finds such a $\tX$, taking time
\[
O(\nnz(A) + \nnz(B) + (n+d+{d'})\poly(r/\eps)) + \tau(d, \poly(r/\eps),\poly(r/\eps), r,\eps).
\]
\end{theorem}

A norm that is \piroi\ 
satisfies the conditions of the theorem, using Lemmas~\ref{lem norm proj}
and \ref{lem UI->embed}. The $v$-norm for $p<1$ also satisfies the conditions
of the theorem, as noted in Remarks~\ref{rem reduced norms} and \ref{rem v-norm embed}.
\Ken{Don't have good examples of right UI but not left UI and not v-norm}

Although earlier results for constrained least squares (e.g. \cite{cw13}) can be applied to obtain
approximation algorithms for regularized multiple-response least squares, via the solution of
$\min_{X\in\R^{d\times d'}} \norm{AX-B}_F^2$, subject to $f(X)\le C$ for a chosen constant $C$,
such a reduction
yields a slower algorithm if properties of $f(X)$ are not exploited, as here.

\begin{proof}
Let $S\in\R^{m_S\times n}$ be an affine embedding as in Lemma~\ref{lem rowspace SB};
that lemma implies
\[
Z^*\equiv\argmin_Z \norm{AZSB - B}_F^2 + f(ZSB)
\]
has
\[
\norm{(AZ^*SB -B)}_F^2 + f(Z^*SB)
	\le (1+\eps)\Delta_*.
\]
Now suppose $\hR\in\R^{{d'}\times m_R}$ comes from a sketching distribution yielding
a right subspace $\eps$-embedding with respect to the Euclidean norm for $SB$,
so that by Lemma~\ref{lem UI->embed} and hypothesis, $\hR$ is a subspace embedding
on the right for $SB$ with respect to $f()$. Suppose also that 
$\hR^\top$ is an affine embedding for $(SB)^\top, B^\top$ with respect to the Frobenius norm.
For example a sparse embedding with $m_R=O(\rank(SB)^2/\eps^2)$ satisfies these
conditions with constant probability.

Suppose $\hS$ is an affine embedding for $A,B\hR$. Then
\begin{equation}\label{eq mr gen sk}
\tZ \equiv \argmin_Z \norm{\hS AZSB\hR - \hS B\hR}_F^2 + f(ZSB\hR)
\end{equation}
has
\[
\norm{(A\tZ SB -B)}_F^2 + f(\tZ SB)
	\le (1+\eps)^3\Delta_*,
\]
so that $\tX \equiv \tZ SB$ satisfies the conditions of the theorem,
up to a constant factor in $\eps$.

We need to put \eqref{eq mr gen sk} into the form of \eqref{eq ABgf exact}.
Let $D\equiv SB \hR$, and let $Q$ have
$m_Q\equiv \rank(D)$ orthogonal columns and $m_R$ rows,
such that for upper triangular $T\in\R^{m_Q\times m_Q}$ and $T'\in\R^{m_Q\times (m_R- m_Q)}$,
$D^\top = Q[T\,\, T']$.
Then any
$ZSB \hR\in  \rowspan(SB\hR)$ can be written as $Z_1Q^\top$,
for some $Z_1\in \R^{d\times m_Q}$. (We can recover $Z$ as in Lemma~\ref{lem boyd},
with a back-solve on $Z_1$ using $T$.)

Letting $P_Q \equiv QQ^\top$, and using $P_Q(\Iden - P_Q)=0$
and matrix Pythagoras,
\eqref{eq mr gen sk} can be solved by minimizing
\begin{align*}
\norm{\hS AZ_1 Q^\top- \hS B\hR}_F^2 + f(Z_1Q^\top)
	    & = \norm{\hS AZ_1 Q^\top P_Q -  \hS B\hR P_Q  + \hS B\hR(P_Q - \Iden) }_F^2 + f(Z_1 Q^\top)
	\\  & = \norm{\hS AZ_1 Q^\top P_Q  -  \hS B\hR P_Q }_F^2 + \norm{(P_Q-\Iden)\hS B\hR}_F^2 + f(Z_1),
\end{align*}
with respect to $Z_1$, using also padding invariance and orthogonal invariance of $f()$.
We could equivalently minimize
\[
	\norm{\hS AZ_1 Q^\top   -  \hS B\hR QQ^\top }_F^2  + f(Z_1)
	= \norm{\hS AZ_1   -  \hS B\hR Q}_F^2  + f(Z_1),
\]
which has the form of \eqref{eq ABgf exact}.

It remains to determine the sketching dimensions for $S$, $\hS$, and $\hR$.
We need $S\in\R^{m_S\times n}$  and $\hS\in\R^{m_{\hS}\times n}$ to be affine embeddings
for $A$, $B$ and for $A$, $B\hR$ with respect to the Frobenius norm.
Sparse embeddings (Def.~\ref{def subs embed}, Lemma~\ref{lem AE})
have this property, with constant probability for $m_S, m_{\hS}=O(r^2/\eps^2)$, where again $r\equiv\rank(A)$.
By hypothesis, we have a distribution over $\hR$ with $m_{\hR}  = \poly(m_S/\eps) = \poly(r/\eps)$ with the needed properties.
Thus the algorithm of the theorem statement would be called with $\tau(d, m_{\hS}, m_{\hR}, r,\eps)$, with
the appropriate parameters in $\poly(r/\eps)$, as claimed.
\end{proof}

\section{General Regularization: Low-rank Approximation}\label{sec low-rank gen}

For an integer $k$ we consider the problem 
\begin{equation}\label{eq lowr gen}
\min_{\substack{Y\in\R^{n\times k}\\ X\in\R^{k\times d}}}
		\norm{YX - A}_F^2 + f(Y,X),
\end{equation}
where $f(\cdot,\cdot)$ is a real-valued function that is
\piloi\ in the left argument, \piroi\ in the right argument,
and left and right reduced by contraction in its left and right arguments, respectively.

For example $\hat f(\norm{Y}_\ell, \norm{X}_r)$ for \piloi\ $\nm{\ell}$ and \piroi\ $\nm{r}$ would satisfy these conditions,
as would $\norm{YX}_g$ for orthogonally invariant norm $\nm{g}$. The function $\hat f$ could be zero for arguments whose
maximum is less than some $\mu$, and infinity otherwise.

\subsection{Via the SVD}\label{subsec low-rank gen svd}

First, a solution method relying on the singular value decomposition for a slightly more general problem
than \eqref{eq lowr gen}.

\begin{theorem}\label{thm reduce kxk}
Let $k$ be a positive integer,
$f_1:\R\mapsto \R$ increasing,
and $f: \R^{n\times k}\times \R^{k\times d} \mapsto \R$,
where $f$ is \piloi\ and left reduced by contractions in its left argument,
and \piroi\ and right reduced by contractions in  in its right argument.
Let $A$ have full SVD $A=U\Sigma V^\top$, $\Sigma_k\in\R^{k\times k}$
the diagonal matrix of top $k$ singular values of $A$.
Let matrices $W^*,Z^*\in\R^{k\times k}$ solve
\begin{equation}
\min_{\substack{W\in\R^{k\times k}\\ Z\in\R^{k\times k} \\ WZ\  \mathrm {diagonal}}}
		f_1(\norm{WZ - \Sigma_k}_{(p)}) + f(W,Z),
\end{equation}
and suppose there is a procedure taking $\tau(k)$ time to find $W^*$ and $Z^*$.
Then the solution to
\begin{equation}\label{eq lowr generer}
\min_{\substack{Y\in\R^{n\times k}\\ X\in\R^{k\times d}}}
		f_1(\norm{YX - A}_{(p)}) + f(Y,X)
\end{equation}
is $Y^*=U\twomat{W^*}{0_{(n-k)\times k}}$ and $X^* = \twomatr{Z^*}{0_{k\times (d-k)}}V^\top$.
Thus for general $A$, \eqref{eq lowr generer} can be solved
in time $O(nd\min\{n,d\}) + \tau(k)$.
\end{theorem}

We will need a lemma.

\begin{lemma}[\cite{chatterjee2014matrix}, Thm 8.1]\label{lem Wie}
Let $A,B\in\R^{n\times d}$, $C\equiv A-B$, and vectors of singular values (in nonincreasing order)
$\sigma_A$, $\sigma_B$, $\sigma_C$. For any $p\in [1,\infty]$,
$\norm{\sigma_A - \sigma_B}_p \le \norm{\sigma_C}_p$.
\end{lemma}

Note that $\norm{\sigma_A}_p$ is the Schatten $p$-norm $\norm{A}_{(p)}$.

%So \eqref{eq lowr gen} in the diagonal case is
%$\min_{u,v} \sum_{i\in [r]} (u_i v_i - \sigma_i)^2 + f(g(u), g(v))$, where $u$ and $v$ have $k$ nonzero entries.

\begin{proof}[Proof of Thm \ref{thm reduce kxk}]
%We can assume that $A$ is square: if $A$ is not square, say $n < d$,
%then padding $A$ with $n-d$ zero rows yields a square matrix $\twomat{A}{0_{(n-d)\times d}}$. For given
%$Y,X$, $\twomat{Y}{0_{(n-d)\times k}}$ 
%has $n$ rows, and $\norm{\twomat{A}{0} - \twomat{Y}{0}X}_F = \norm{A-YX}_F$.
%(Here we omit the row and column sizes for the zero paddings.)
%Since $f(\cdot,\cdot)$ is padding invariant in each argument,
%$(\twomat{Y}{0},X)$ has the same cost with respect to $\twomat{A}{0}$ as $(Y,X)$ have
%with respect to $A$, and any $(\twomat{Y}{\hat Y}, X)$ has at least the cost for $\twomat{A}{0}$ as $(Y,X)$ do for $A$,
%since the mapping from $\twomat{Y}{\hat Y}$ to $\twomat{Y}{0}$ is a contraction. So
%indeed, we can assume that $A$ is square.

Suppose $A$ has full SVD $A=U\Sigma V^\top$, and $U^\top YX V$ has full SVD $RDS^\top$,
and let $W\equiv R^\top U^\top Y$ and $Z\equiv XVS$, so that $WZ = D$.
Then the invariance properties of $\nm{(p)}$ and $f(\cdot,\cdot)$ imply
\begin{align*}
f_1(\norm{YX - A}_{(p)}) + f(Y,X)
	   & = f_1(\norm{U R W Z S^\top V^\top - U \Sigma V^\top}_{(p)}) + f(U R W,  Z S^\top V^\top)
	\\ & = f_1(\norm{R W Z S^\top - \Sigma}_{(p)}) + f(W, Z) 
	\\ & = f_1(\norm{RDS^\top - \Sigma}_{(p)}) + f(W,Z).
\end{align*}
So the objective function is no larger at $W, Z$ than at $Y,X$ if
$\norm{WZ - \Sigma}_{(p)} = \norm{D-\Sigma}_{(p)}\le \norm{RDS^\top - \Sigma}_{(p)}$.
We apply Lemma~\ref{lem Wie}, with $A$ of the lemma mapped to $RDS^\top$ and $B$ to $\Sigma$,
and use the the relation of the Schatten norm to the vector $p$-norm.
The bound follows, 
%Using orthogonal invariance of $\nm{{(p)}}$ we have
%\[
%\norm{RDS^\top - \Sigma}_F^2 - \norm{D-\Sigma}_F^2
%	 = - 2\tr(\Sigma^\top RDS^\top) + 2\tr(\Sigma^\top D).
%\]
%As noted above, we can assume that $A$ is square, so that the Von Neumann trace inequality
%can be applied; it implies that
%$\tr(\Sigma^\top RDS^\top) \le \tr(\Sigma^\top D)$,
%and so $\norm{D-\Sigma}_F^2 \le \norm{RDS^\top - \Sigma}_F^2$,
and we can assume that $WZ$ is a diagonal matrix $D$.

Since $D$ has rank at most $k$, it has at most $k$ nonzero entries; we will assume
$\rank(D)=k$, but similar arguments go through for $\rank(D)<k$.
Let $P_D$ have ones where $D$ is nonzero,
and zeros otherwise. Then $P_DW$ is the projection of $W$ onto the rowspace of $D$, and $Z P_D$ is the 
projection of $Z$ onto $D$'s columnspace. Since $f(\cdot,\cdot)$ is appropriately reduced by
contractions, and $P_DWZ P_D = D$,
we can assume that all but at most $k$ rows of $W$ and columns of $Z$ are zero. Removing these
zero rows and columns, we have $k\times k$ matrices $D$, $W$, $Z$, and $\Sigma$,
and $W$ and $Z$ are invertible. \Ken{yes, invertible?}
(Here we use padding invariance, but only to extend $f$ to smaller matrices.)

Since the rows of $W$ can be swapped by multiplying by an orthogonal matrix on the left,
and the columns of $Z$ via an orthogonal matrix on the right, the nonzero entries of $D=WZ$
can be moved to correspond to the $k$ largest diagonal entries of $\Sigma$ without changing $f(W,Z)$,
and such moves can only decrease $\norm{D-\Sigma}_{(p)}$.
\end{proof}

We sharpen this result for the case that the regularization term comes from
orthogonally invariant norms.

\begin{theorem}\label{thm reduce kxk diag}
Consider \eqref{eq lowr generer} when $f(\cdot,\cdot)$ has the form $\hat f(\norm{Y}_\ell, \norm{X}_r)$, where 
$\nm{\ell}$ and $\nm{r}$ are orthogonally invariant, and $\hat f :\R\times \R \mapsto \R$ increasing in each argument.
Suppose in that setting there is a procedure that solves \eqref{eq lowr generer} when $A$, $Y$, and $X$ are diagonal matrices, taking time
$\tau(r)$ for a function $\tau(\cdot)$, with $r\equiv\rank(A)$.
Then for general $A$, \eqref{eq lowr generer} can be solved by finding the 
SVD of $A$, and applying the given procedure to $k\times k$ diagonal matrices, taking altogether time 
$O(nd\min\{n,d\}) + \tau(k)$.
\end{theorem}

We will need a lemma.

\begin{lemma}\label{lem R inc}
If $E,D,R\in\R^{n\times n}$ with $D$ and $E$ diagonal, and $R$ orthogonal, for any orthogonally invariant norm $\nm{g}$, there is a permutation $\pi$ on $[n]$ so tha
 $\norm{\pi(E) D}_g \le \norm{E R D}_g$, where $\pi(E)_{i,i} \equiv E_{\pi(i),\pi(i)}$.
\end{lemma}

\begin{proof}
The permutation $\pi$ we choose is the one that puts the $i$'th largest entry of $|E|$ with the $i$'th smallest entry of $|D|$.
Since the singular values of $E$ and $D$ are the nonzero entries of $|E|$ and $|D|$, this means that the singular values of $\pi(E)D$
have the form $\sigma_i(E)\sigma_{n-i+1}(D)$, where $\sigma_i(\cdot)$ denotes the $i$'th largest singular value.
We use an inequality of \cite{wang1992some}, page 117, which implies that for any $k\in [n]$
and ${\cal S}\subset [n]$ of size $k$,
$\sum_{i\in [k]} \sigma_i(ERD) \ge \sum_{i\in\cal S} \sigma_i(E)\sigma_{n-i+1}(D)$. Since $\cal S$ can be the set of
indices of the $k$ largest entries of $|\pi(E)|*|D|$, which are the $k$ largest singular values of $\pi(E)D$,
this implies that for all $k$, the sum of the $k$ largest singular values of $ERD$ is larger than the corresponding sum
for $\pi(E)D$.  Therefore by the Ky Fan dominance theorem \cite{fan1951maximum}, the lemma follows.
\end{proof}

\begin{proof}[Proof of Thm~\ref{thm reduce kxk diag}]
Following up on the proof of Theorem~\ref{thm reduce kxk},
it suffices to show that when $\nm{\ell}$ and $\nm{r}$ are orthogonally invariant,
it can be assumed that $W$ and $Z$ are diagonal matrices.

Let $W$ have the SVD $W=U_{W} \Sigma_{W} V_{W}^\top$. Then
$Z = W^{-1}D = V_{W} \Sigma_{W}^{-1} U_{W}^\top D$,
so that $\hat f(\norm{W}_\ell, \norm{Z}_r) = \hat f(\norm{\Sigma_{W}}_\ell, \norm{\Sigma_{W}^{-1} U_{W}^\top D}_r)$, using orthogonal
invariance. We now apply Lemma~\ref{lem R inc}, with $E$ of the lemma mapping to $\Sigma_{W}^{-1}$, $R$ to
$U_{W}^\top$, and $D$ to $D$. This yields a permutation $\pi$ on the entries of $\Sigma_{W}^{-1}$ so that
$\norm{\pi(\Sigma_{W}^{-1}) D}_r \le \norm{\Sigma_{W}^{-1} U_{W}^\top D}_r$, so that
the diagonal matrices $\pi(\Sigma_{W})$ and $\pi(\Sigma_{hW}^{-1})D$ have product $D$
and objective function value no larger than $W$ and $Z$; that is, without loss of generality,
$W$ and $Z$ are diagonal. Thus minimizing after obtaining the singular values $\Sigma$ of $A$,
the solution of $\norm{WZ - \Sigma}_F^2 + \hat f(\norm{W}_\ell, \norm{Z}_r)$ with $W$ and $Z$ diagonal
is sufficient to solve \eqref{eq lowr generer}.
\end{proof}

\begin{definition}[clipping to nonnegative $(\cdot)_+$]
For real number $a$, let $(a)_+$ denote $a$, if $a\ge 0$, and zero otherwise. For matrix $A$,
let $(A)_+$ denote coordinatewise application.
\end{definition}

\begin{corollary}\label{cor simple sol}
If the objective function in \eqref{eq lowr generer} is $\norm{YX-A}_F^2 + 2\lambda\norm{YX}_{(1)}$
or $\norm{YX-A}_F^2 + \lambda(\norm{Y}_F^2 +\norm{X}_F^2)$, then the diagonal matrices
$W^*$ and $Z^*$ from Theorem~\ref{thm reduce kxk diag} yielding the solution
are $W^*=Z^*=\sqrt{(\Sigma_k - \lambda\Iden_k)_+}$, where $\Sigma_k$ is the $k\times k$ diagonal matrix
of top $k$ singular values of $A$ \cite{UHZB}. 

If the objective function is $\norm{YX-A}_{(p)}+ \lambda\norm{YX}_{(1)}$ for $p\in [1,\infty]$,
then $W^*=Z^*=\sqrt{(\Sigma_k - \alpha\Iden_k)_+}$, for an appropriate value $\alpha$.

If the objective function is $\norm{YX-A}_F^2 + \lambda\norm{YX}_F^2$, then
$W^*=Z^*=\sqrt{\Sigma_k/(1+\lambda)}$.

\Ken{ etc. I think these are right. too bad they're boring.}
\end{corollary}

\begin{proof} Omitted.
\end{proof}

\subsection{Reduction to a small problem via sketching}

\begin{theorem}\label{thm lowr gen sk}
Suppose there is a procedure that solves \eqref{eq lowr gen} when $A$, $Y$, and $X$ are $k\times k$ matrices,
and $A$ is diagonal, and $YX$ is constrained to be diagonal, taking time
$\tau(k)$ for a function $\tau(\cdot)$.
Let $f$ also inherit a sketching distribution on the left in its left argument, and on the right in its right argument.
Then for general $A$, there is an algorithm that finds $\eps$-approximate solution $(\tY, \tX)$ in time
\[
O(\nnz(A)) + \tO(n+d)\poly(k/\eps)  + \tau(k).
\] 
\end{theorem}

\begin{proof}
We follow a sequence of reductions similar to those for Theorem~\ref{thm gen regul regr}, but on both sides.

Let $(Y^*, X^*)$ be an optimal solution pair:
\begin{equation}\label{eq lowr gener}
Y^*,X^*\equiv \argmin_{\substack{Y\in\R^{n\times k}\\ X\in\R^{k\times d}}} \norm{YX-A}_F^2 + f(Y,X),
\end{equation}
and $\Delta_* \equiv \norm{Y^*X^*-A}_F^2 + f(Y^*, X^*)$.

Let $S\in\R^{m_S\times n}$ be an affine $\eps$-embedding for $Y^*, A$ with respect to $\nm{F}$.
From Lemma~\ref{lem rowspace SB}, 
\[
Z^*\equiv\argmin_{Z\in\R^{k\times m_S}} \norm{Y^*ZSA - A}_F^2 + f(Y^*, ZSA)
\]
has
\[
\norm{(Y^*Z^*SA - A)}_F^2 + f(Y^*, Z^*SA)
	\le (1+\eps)\Delta_*.
\]
Now suppose $R\in\R^{d\times m_R}$ is a right affine $\eps$-embedding for $Z^*SA, A$ with respect
to $\nm{F}$. Then again by Lemma~\ref{lem rowspace SB}, applied on the right,
\[
W^* \equiv \argmin_{W\in\R^{m_R\times k}} \norm{ARW Z^* SA - A}_F^2 + f(ARW, Z^*SA)
\]
has
\[
\norm{(AR W^* Z^*SA - A)}_F^2 + f(ARW^*, Z^*SA)
	\le (1+\eps)^2\Delta_*.
\]
It doesn't hurt to find the best $W^*, Z^*$ simultaneously, so redefining them to be
\begin{equation}\label{eq bi-solve}
W^*, Z^* \equiv \argmin_{\substack{W\in\R^{m_R\times k}\\ Z\in\R^{k\times m_S}} }
	\norm{ARWZSA - A}_F^2 + f(ARW, ZSA)
\end{equation}
satisfies the same approximation property.

Suppose $\hR\in \R^{d\times m_{\hR}}$ comes from a sketching distribution yielding a right subspace $\eps$-embedding
with respect to the Euclidean norm for $SA$, so that by assumption, $\hR$ is a subspace $\eps$-embedding
on the right for $SA$ with respect to the right argument of $f(\cdot,\cdot)$.
Suppose also that $\hR^\top$ is an affine embedding for $(Z^*SA)^\top, A^\top$ with respect to the Frobenius norm.
Suppose $\hS$ is similarly a left subspace $\eps$-embedding for $AR$ with respect to the left argument
of $f(\cdot,\cdot)$,
and an affine
embedding on the left for $AR\tW, A\hR$ with respect to the Frobenius norm,
where $\tW$ is the solution to
$\min_{W\in\R^{m_R\times k}} \norm{ARW Z^* SA\hR - A\hR}_F^2 + f(ARW, Z^*SA\hR)$.
Then
\begin{equation}\label{eq lowr reduc}
\tW, \tZ \equiv \argmin_{\substack{W\in\R^{m_R\times k}\\ Z\in\R^{k\times m_S}} }
	\norm{\hS AR WZ SA\hR - \hS A\hR}_F^2 + f(\hS ARW, ZSA\hR)
\end{equation}
form a $(1+O(\eps))$-approximate solution to \eqref{eq bi-solve}, and therefore yield a $(1+O(\eps))$-approximate
solution to \eqref{eq lowr gener}.

We need to put the above into the form of \eqref{eq lowr gen}. Suppose $Q_\ell$ is an orthogonal basis
for $\colspace(\hS AR)$, and $Q_r^\top$ an orthogonal basis for $\rowspan(SA\hR)$.
Then any matrix of the form $\hS ARW$ can be written as $Q_\ell W_1$ for some $W_1\in\R^{\rank(SA\hR)\times k}$,
and similarly any matrix of the form $Z SA\hR$ can be written as $Z_1 Q_r^\top$ for
some $Z_1$. Thus solving \eqref{eq lowr reduc} is equivalent to solving
\[
\tW_1, \tZ_1 \equiv \argmin_{W_1, Z_1 }
	\norm{Q_\ell W_1 Z_1 Q_r^\top  - \hS A\hR}_F^2 + f(Q_\ell W_1 , Z_1 Q_r^\top).
\]
(We can recover $\tW$ and $\tZ$ from $\tW_1$ and $\tZ_1$ via back-solves with the 
triangular portions of change-of-basis matrices, and padding by zeros, as in Lemma~\ref{lem boyd}
and Theorem~\ref{thm gen regul regr}.)
Using the properties of $f(,)$ we have $ f(Q_\ell W_1 , Z_1 Q_r^\top) = f(W_1, Z_1)$.
Let $P_\ell\equiv Q_\ell Q_\ell^\top$, and $P_r \equiv Q_r Q_r^\top$.
Using $P_\ell(\Iden - P_\ell) = 0$ and $P_r(\Iden - P_r) = 0$ and matrix Pythagoras, we have
\begin{align*}
\norm{Q_\ell W_1 Z_1 Q_r^\top  - \hS A\hR}_F^2 & + f(Q_\ell W_1 , Z_1 Q_r^\top)
	\\ & = \norm{P_\ell Q_\ell W_1 Z_1 Q_r^\top P_r  - \hS A\hR}_F^2 +  f(W_1, Z_1)
	\\ & = \norm{P_\ell Q_\ell W_1 Z_1 Q_r^\top P_r  - P_\ell \hS A\hR P_r}_F^2 
	\\ & \qquad + \norm{(\Iden - P_\ell)  \hS A\hR}_F^2
		+ \norm{ P_\ell  \hS A\hR (\Iden - P_r) }_F^2
		+ f(W_1, Z_1)
\end{align*}
So we could equivalently minimize
\begin{align*}
\norm{P_\ell Q_\ell W_1 Z_1 Q_r^\top P_r  & - P_\ell \hS A\hR P_r}_F^2 
		+ f(W_1, Z_1)
	\\ &  = \norm{Q_\ell W_1 Z_1 Q_r^\top - Q_\ell Q_\ell^\top \hS A\hR Q_r Q_r^\top}_F^2 + f(W_1, Z_1)
	\\ & =  \norm{W_1 Z_1 - Q_\ell^\top \hS A\hR Q_r}_F^2 + f(W_1, Z_1),
\end{align*}
which has the form of \eqref{eq lowr gen}.

%
%
%
%Let $\hS AR$ have the full SVD $\hS AR=U_\ell \Sigma_\ell V_\ell^\top$,
%and $SA\hR$ the full SVD $SA\hR=U_r\Sigma_r V_r^\top$. Using the invariance properties
%of $\nm{F}$ and $f(\cdot,\cdot)$, we have
%\[
%\tW, \tZ = \argmin_{\substack{W\in\R^{m_R\times k}\\ Z\in\R^{k\times m_S}} }
%	\norm{\Sigma_\ell V_\ell^\top WZU_r\Sigma_r - U_\ell^\top \hS A\hR V_r }_F^2 + f(\Sigma_\ell V_\ell^\top W, Z U_r\Sigma_r).
%\]
%
%As in the proof of Theorem~\ref{thm gen regul regr}, since the zero rows of $\Sigma_\ell$ and the zero columns
%of $\Sigma_r$ are preserved in the product, we can multiply the expression within the
%$\nm{F}$ on the left by $\Iden_{m_R\times m_{\hS}}$,
%and on the right by $\Iden_{m_{\hR}\times m_S}$, and obtain the same solution.
%Also as in the proof of Theorem~\ref{thm gen regul regr}, we assume WLOG
%that $\hS AR$ and $SA\hR$ have full column and row rank, so that $\Iden_{m_R\times m_{\hS}}\Sigma_\ell$
%and $\Sigma_r\Iden_{m_{\hR}\times m_S}$ are invertible. That is, we substitute
%$\hW \equiv \Iden_{m_R\times m_{\hS}}\Sigma_\ell V_\ell^\top W$, and
%$\hZ \equiv Z U_r \Sigma_r \Iden_{m_{\hR}\times m_S}$, and solve
%\[
%\min_{\substack{\hW\in\R^{m_R\times k}\\ \hZ\in\R^{k\times m_S}} } 
%	\norm{\hW\hZ - \Iden_{m_R\times m_{\hS}} U_\ell^\top \hS A\hR V_r \Iden_{m_{\hR}\times m_S}}_F^2 + f(\hW,\hZ),
%\]
%then solve for $\tW, \tZ$, and therefore for a solution to \eqref{eq lowr gener}.

It remains to determine the sizes of $S$, $R$, $\hR$, and $\hS$, and the cost of their applications.
We use the staged construction
of Lemma~\ref{lem AE}, so each of these matrices is the product of a sparse embedding and an SHRT.
We have $m_R$ and $m_S$ both $\tO(k/\eps^2)$, and $m_{\hR}=m_{\hS}=\tO(k/\eps^4)$,
noting that we need $\hS$ to be a subspace $\eps$-embedding for $AR$, of rank $\tO(k/\eps^2)$,
and similarly for $\hR$. Moreover,
to compute $\hS AR$, $S A\hR$, and $\hS A \hR$, we can first apply the sparse embeddings 
on either side, and then the SHRT components, so that the cost of computing these
sketches is $O(\nnz(A)) + \tO(k^2/\eps^6)$. Since the remaining operations involve matrices
with $\tO(k/\eps^4)$ rows and columns, the total work, up to computing
$AR\tW$ and $\tZ SA$, is $O(\nnz(A)) + \tO(\poly(k/\eps)) + \tau(k)$. The work to compute those 
products is $O(n+d)\poly(k/\eps)$, as claimed.
\end{proof}

% !TEX root = p.tex

\section{Estimation of statistical dimension}\label{sec sd est}

\Ken{"In Progress"; as written is $\nnz(A)\log n$, but should be able to do better with one big sketch}

\begin{theorem}
If the statistical dimension $\sd_\lambda(A)$ is at most
\[
M\equiv \min\{n,d, \lfloor (n+d)^{1/3}/\poly(\log (n+d))\rfloor\},
\]
it can be estimated to within a constant factor in $O(\nnz(A))$ time,
with constant probability.
\end{theorem}

\begin{proof}
From Lemma~18 of \cite{CEMMP}, generalizing the machinery of \cite{AN}, the first $z$ squared singular values of $A$ can be estimated up to additive
$\frac{\eps}{z}\norm{A_{-z}}_F^2$ in time $O(\nnz(A)) + \tO(z^3/\poly(\eps))$, where $A_{-z} \equiv A - A_z$ denotes the residual error of the best rank-$z$ approximation $A_z$ to $A$. Therefore $\norm{A_z}_F^2$ can be estimated up to additive $\eps\norm{A_{-z}}_F^2$,
and the same for $\norm{A_{-z}}_F^2$. This implies that for small enough constant $\eps$, $\norm{A_{-z}}_F^2$ can be estimated up to
constant relative error, using the same procedure.

Thus in $O(\nnz(A))$ time, the first $6M$ singular values of $A$ can be estimated up to additive
$\frac{1}{6M}\norm{A_{-6M}}_F^2$ error, and there is an estimator $\hat\gamma_z$ 
of $\norm{A_{-z}}_F^2$ up to relative error $1/3$, for $z\in[6M]$.

Since $1/(1+\lambda/\sigma_i^2) \le \min\{1, \sigma_i^2/\lambda\}$, for any $z$ the summands of $\sd_\lambda(A)$ for $i\le z$ are at most 1,
while those for $i > z$ are at most $\sigma_i^2/\lambda$, and so $\sd_\lambda(A) \le z + \norm{A_{-z}}_F^2/\lambda$.

When $\sigma_z^2 \le \lambda$, the summands of $\sd_\lambda(A)$ for
$i\ge z$ are at least $\frac12\frac{\sigma_i^2}{\lambda}$, and so $\sd_\lambda(A)\ge \frac12\norm{A_{-z}}_F^2/\lambda$.
When $\sigma_z^2 \ge \lambda$,
the summands of $\sd_\lambda(A)$ for $i\le z$ are at least $1/2$. Therefore 
$\sd_\lambda(A)\ge \frac12\min\{z, \norm{A_{-z}}_F^2/\lambda\} $.

Under the constant-probability assumption that $\hat\gamma_z=(1\pm 1/3)\norm{A_{-z}}_F^2$, we have
\begin{equation}\label{eq sd bounds}
\frac38\min\{z, \hat\gamma_z/\lambda\} \le \sd_\lambda(A) \le \frac32(z + \hat\gamma_z/\lambda).
\end{equation}

Let $z'$ be the smallest $z$ of the form $2^j$ for $j=0,1,2,\ldots$, with $z'\le 6M$, such that $z' \ge \hat\gamma_{z'}/\lambda$.
Since $M\ge \sd_\lambda(A) \ge \frac38 z$ for $z\le \hat\gamma_z/\lambda$, there must be such a $z'$.
Then by considering the lower bound of \eqref{eq sd bounds} for $z'$ and for $z'/2$, we have
$\sd_\lambda(A) \ge \frac38\max\{z'/2, \hat\gamma_{z'}/\lambda\}\ge \frac1{16}(z' + \hat\gamma_{z'}/\lambda)$,
which combined with the upper bound of \eqref{eq sd bounds} implies that
$z' + \hat\gamma_{z'}/\lambda$ is an estimator of $\sd_\lambda(A)$ up to a constant factor.
\end{proof}

\ifJOURNAL
\bibliographystyle{ACM-Reference-Format-Journals}
\fi

\ifJOURNAL
% History dates
\received{November 2013}{0}{0}

% Electronic Appendix
%\elecappendix
%\appendix
\else\ifSUB
%\appendix
\section{Omitted Proofs}
\subsection{Proof of Lemma~\ref{lem reg}}\label{subsec lemregproof}
\lemregproof
\subsection{Proof of Lemma~\ref{lem U1 size}}\label{subsec lemUonesizeproof}
\lemUonesizeproof
\subsection{Proof of Lemma~\ref{lem lam large}}\label{subsec lemlamlargeproof}
\lemlamlargeproof
\subsection{Proof of Corollary~\ref{cor size of S}}\label{subsec corsizeofSproof}
\corsizeofSproof
\subsection{Proof of Theorem~\ref{thm reg stacked}}\label{subsec thmregstacked}
\thmregstackedproof
\subsection{Proof of Theorem~\ref{thm:twoProp}}\label{subsec thmtwoPropproof}
\thmtwoPropproof

\fi

\section*{Acknowledgments}

The authors acknowledge the support from the XDATA program of the Defense
Advanced Research Projects Agency (DARPA), administered through Air
Force Research Laboratory contract FA8750-12-C-0323.

\bibliographystyle{alpha}
\bibliography{p}

\end{document}